\documentclass[11pt]{amsart}

\usepackage[T1]{fontenc}
\usepackage[utf8]{inputenc}
\usepackage[english]{babel} 
\usepackage{amsmath,amssymb,amsfonts,amsthm}
\usepackage{mathrsfs, esint, dsfont}
\usepackage{moreverb,rotating,graphics,float}
\usepackage{caption, subcaption}
\usepackage[colorlinks, linkcolor={blue}, citecolor={red}]{hyperref}
\usepackage{framed,fancybox}
\usepackage{enumerate}
\usepackage{slashed}
\usepackage{cancel}
\usepackage{color}
\usepackage{xcolor}
\usepackage[normalem]{ulem}

\def\fm{{\mathfrak{m}}}
\renewcommand{\epsilon}{\varepsilon}

\newcommand{\trv}{{\rm \underline{Tr}\,}}

\newcommand{\cK}{{\mathcal K}}
\newcommand{\cG}{{\mathcal G}}

\newcommand{\cW}{{\mathcal W}}
\newcommand{\cV}{{\mathcal V}}

\newcommand\B{{\mathcal B}}

\newcommand\G{{\mathcal G}}

\newcommand\NN{{\mathbb N}}

\newcommand\RR{{\mathbb R}}

\newcommand\CC{{\mathbb C}}

\newcommand\ii{{\infty}}
\newcommand\1{{\mathds{1}}}

\newcommand{\norm}[1]{\left\| #1\right\|}
\newcommand{\set}[1]{\left\{ #1\right\}}
\newcommand{\bra}[1]{\left( #1\right)}
\newcommand{\av}[1]{\left| #1\right|}
\newcommand{\com}[1]{\left[ #1\right]}

\renewcommand{\phi}{\varphi}

\def\bP{\mathbf{P}}
\def\bM{\mathbf{M}}
\def\bA{\mathbf{A}}

\def\bB{\mathbf{B}}

\def\bK{\mathbf{K}}
\def\bL{\mathbf{L}}

\def\bp{\mathbf{p}}

\def\bR{\mathbf{R}}

\def\bS{\mathbf{S}}

\def\bx{\mathbf{x}}
\def\by{\mathbf{y}}

\def\btheta{\boldsymbol{\theta}}

\def\div{{\rm div \;}}
\def\NN{{\mathbb N}}

\def\RR{{\mathbb R}}
\def\CC{{\mathbb C}}

\newcommand{\Tr}{{\rm Tr}}
\newcommand{\VTr}{\underline{\rm Tr}}

\newcommand{\curl}{{\bf curl} \,}

\newcommand{\rd}{\mathrm{d}}
\newcommand{\re}{\mathrm{e}}
\newcommand{\ri}{\mathrm{i}}

\newcommand{\cH}{\mathcal{H}}

\newcommand{\cF}{\mathcal{F}}

\newcommand{\cS}{\mathscr{S}}

\newcommand{\cE}{\mathcal{E}}

\renewcommand{\leq}{\leqslant}
\renewcommand{\le}{\leqslant}
\renewcommand{\geq}{\geqslant}
\renewcommand{\ge}{\geqslant}

\numberwithin{equation}{section} 
\newtheorem{theorem}{Theorem}[section]
\newtheorem{proposition}[theorem]{Proposition}
\newtheorem{remark}[theorem]{Remark}
\newtheorem{lemma}[theorem]{Lemma}
\newtheorem{corollary}[theorem]{Corollary}

\newcommand{\salma}[1]{{\color{cyan}[Salma: #1]}}
\newcommand{\abdallah}[1]{{\color{teal}[Abdallah: #1]}}
\newcommand{\carlos}[1]{{\color{magenta}[Carlos: #1]}}

\title{Electronic structure models with 2D symmetries in the presence of magnetic fields} 

\author{Carlos J. Garc{\'\i}a-Cervera, Salma Lahbabi and Abdallah Maichine}
\date{\today}

\address[Carlos J. Garc{\'\i}a-Cervera]{Mathematics Department, University of California, Santa Barbara, CA 93106, USA and Basque Center for Applied Mathematics (BCAM), Bilbao, Basque Country, Spain}
\email{cgarcia@ucsb.edu}

\address[Salma Lahbabi]{EMA, LARILE, ENSEM, University of Hassan II in Casablanca, Morocco and Makhbar Mathematical Science Institute, Casablanca, Morocco}
\email{s.lahbabi@ensem.ac.ma}

\address[Abdallah Maichine]{LAMA, Faculty of Sciences, Mohammed V University in Rabat, Morocco}
\email{a.maichine@um5r.ac.ma}

\begin{document}
\begin{abstract}
In this work, we characterize self-adjoint operators that commute with magnetic translations. We use this characterization to derive effective kinetic energy functionals for homogeneous electron gases and three-dimensional electronic systems with two-dimensional symmetries in the presence of a magnetic field.
\end{abstract}

\maketitle
\tableofcontents
\section{Introduction} 

The study of quantum electronic structure models has been a central theme in condensed matter physics for almost a century. In the non-relativistic framework, the ground-state theory of some of these models for atoms and molecules is by now fairly well understood in a mathematical sense. This is the case of the Thomas-Fermi and Thomas-Fermi-von Weizsäcker models \cite{Thomas:27,Fermi:27,Lieb_DFT}, the Hartree-Fock model  \cite{Hartree_1928,fock:1930,Slater:1930,LiebSimon:1977,Lions1987}, and  some Density Functional Theory (DFT) approaches \cite{Kohn,Levy1979,Lieb1983}, in the context of the Kohn-Sham equations \cite{KS_LDA,Anantharaman2009}.

\medskip
For finite systems, the ground state can be characterized as a minimizer of the quantum energy functional. However, since crystals are infinite periodic systems, every state carries infinite energy, and an appropriate definition of a ground state is required. One natural approach is through the thermodynamic limit: considering a finite subsystem, computing its ground-state energy, normalizing by its size, and then taking the limit as the subsystem grows to infinity. 


\medskip

In the absence of magnetic fields, the existence of the thermodynamic limit has been established in various settings: for three-dimensional crystals in the Thomas–Fermi ~\cite{LiebSimon} and Thomas-Fermi-von Weizsäcker models 
\cite{Catto1998_book}; for the reduced Hartree–Fock model for perfect crystals~\cite{CLL_periodic}; for crystals with local defects~\cite{CDL}; and for disordered or stochastic systems~\cite{La-13, CaLaLe_proc-12}. The convergence rate of this process has also been investigated, see for instance~\cite{GL2015,CaLaLe-12,GL2016}. In addition, some  fundamental mathematical properties have been analyzed in~\cite{Solovej-TFW,Solovej-io-HF,Solovej1991}. Even in the non-periodic setting, important progress has been made on the definition of ground state energies for infinite systems~\cite{BLL-03}.

\medskip

In the presence of magnetic fields, the existence of the thermodynamic limit was proven by Hainzl and Lewin~\cite{Hainzl2009_2} as part of a general framework for generalized energy functionals with symmetry invariance. However, their approach is based on the many-body model and therefore does not provide information on the properties of ground states nor on convergence rates. More recently, the derivation of effective mean-field dynamics in magnetic settings was studied in~\cite{Benedikteretal}, where the authors established convergence from the many-body Schr\"odinger equation to a nonlinear Hartree–Fock model.

\medskip

In this article, we derive the kinetic energy per unit surface (in dimension two) and per unit volume (in dimension three) as the thermodynamic limits for non-interacting homogeneous electron gases in the presence of a uniform magnetic field. We also consider the case where the electron density $\rho(x_1,x_2,x_3)$ has 2d symmetry, e.g. $\rho(x_1,x_2,x_3)=\rho(0,0,x_3)$, and rewrite the kinetic energy per unit surface as a one-dimensional energy functional as in \cite{GLM21,GLM23}.

\medskip
More precisely, in the case of homogeneous electron gases, let us consider, without loss of generality, a constant magnetic field $\bB=(0,0,b)$, with $b\ge0$ and let $\bA$ be an associated vector potential ($\curl\bA=\bB$ and $\div \bA=0$). The homogeneous electron gas has a constant density ($\rho={\rm const}$) and its ground state kinetic energy density can be formally defined as
\begin{equation}\label{eq:omega}
	\omega^{2d}(b,\rho)=\lim_{L\to\infty}\frac{1}{L^2}\inf\set{\Tr(\bL_\bA^{2d} \gamma) : \gamma\in\bS(L^2(\Gamma_L^{2})),\; 0\le\gamma\le1,\; \Tr(\gamma)=L^2\rho },
\end{equation}
and 
\begin{equation}
    \omega^{3d}(b,\rho)=\lim_{L\to\infty}\frac{1}{L^3}\inf\set{\Tr(\bL_\bA^{3d} \gamma) : \gamma\in\bS(L^2(\Gamma^{3}_L)),\; 0\le\gamma\le1, \;\Tr(\gamma)=L^3\rho },
\end{equation}
where $\Gamma^{2}_L=[-\frac{L}{2},\frac{L}{2}]^2$,  $\Gamma^{3}_L=[-\frac{L}{2},\frac{L}{2}]^3$ and  $\bL_\bA^{2d}$ and $\bL_\bA^{2d}$ are, respectively, the 2d and 3d Landau operators $\bL_\bA^{jd}:= \sum_{\ell=1}^j(\bp_\ell+A_\ell)^2$, and $\bp_\ell=-i\partial_\ell$ for $j\in\set{2,3}$  (see Section~\ref{sec:landau}). 
The condition $0\leq \gamma \leq 1$ is a reflection of the Pauli principle, which insures that two electrons cannot be in the same state. 
We prove that the kinetic energy per unit surface $\omega^{2d}$ and the kinetic energy per unit volume $\omega^{3d}$ of a homogeneous gas have the explicit expressions 
$$ \omega^{2d}(b,\rho)=\pi\rho^2 +b^2\set{\frac{2\pi\rho}{b}}\left( 1-\set{\frac{2\pi\rho}{b}}\right),$$
where $\set{x}=x-\lfloor x\rfloor$ refers to the fractional part of the real number $x\in\RR$,\\ and
\begin{align*}
&\omega^{3d}(b,\rho)=\frac{\delta\rho}{3}+\frac{b^2}{3\pi^2}\sum_{n\in\NN_0} \epsilon_n^{b}\left( \delta -\epsilon_n^{b}\right)_+^{1/2},
\end{align*}
where $\delta=\delta(b,\rho)$ is the Fermi level (chemical potential) determined by the charge constraint. 
We note that $\omega^{3d}(b,\rho)$ appears in the magnetic Thomas-Fermi energy functional as a substitute of the classical kinetic energy density $C_{\rm TF} \rho^{5/3}$, see for instance \cite{LSY94-1,LSY94-2} and \cite{MadSor20}. 

Our method of proving the above statement begins with rewriting $\omega^{jd}(b,\rho)$ as a minimization problem on density matrices that commute with magnetic translations $\{\fm_\bR^\bB\}_{\bR\in\RR^j}$, 
a family of suitable unitary operators (see Section~\ref{sec:magnetic-translations}). Then the main proof ingredient is a characterization of states that commute with the family $\{\fm_\bR^\bB\}_{\bR\in\RR^j}$. This is given in  Theorems~\ref{thm decomposition of 2d operators for B orthogonal} and~\ref{thm decomposition of 3d operators for B orthogonal} for the two-dimensional and three-dimensional operators, respectively.  To the best of our knowledge, this result is new to date. Actually, for ordinary translations $(\tau_\bR)_{\bR\in\RR^d}$, corresponding to $\bB=0$, it is well-known that the  operators invariant under all translations are Fourier multipliers; they can be written as $f(\nabla)$ for measurable functions $f:\RR^d\to\CC$, since $\nabla$ constitutes the 'multi-generator' of $(\tau_\bR)_{\bR\in\RR^d}$. Such an argument cannot apply to magnetic translations as they do not form a group (see Section~\ref{sec:magnetic-translations}), and their 'multi-generators' do not commute among themselves. Our result also allows us to compute explicitly the trace per unit area (surface or volume) for operators commuting with magnetic translations. It reads as follows: if $\gamma$ is a self-adjoint operator on $L^2(\RR^2)$, with locally finite trace, then it can be written in the form
\begin{equation}\label{eq intro gamma=sum ...}
    \gamma=\sum_n \lambda_n\bK_{\psi_n}^{2d},
\end{equation}
for some $\ell^1$-sequence $(\lambda_n)_n$ and an orthonormal basis $(\psi_n)_n$ of $L^2(\RR)$. For every $n$, $K_{\psi_n}^{2d}$ is an infinite dimensional orthogonal projector onto a suitable subspace $E_{\psi_n}$ of $L^2(\RR^2)$, see Theorem~\ref{thm decomposition of 2d operators for B orthogonal} for further details. The construction of $E_{\psi}$'s, for $\psi\in L^2(\RR)$, is done through a Wigner type transform, fixing $\psi$ as a window (see Section~\ref{sec Wigner transform}).
In the 3d setting, a self-adjoint $\gamma$ on $L^2(\RR^3)$ with locally finite trace that commutes with 2d magnetic translations has a similar form as in~\eqref{eq intro gamma=sum ...}, where the spectral projectors $\bK_{\psi}^{3d}$ are orthogonal projectors onto analogous subspaces of $L^2(\RR^2)$, see Theorem~\ref{thm decomposition of 3d operators for B orthogonal}. 

\medskip 
The last major result of this article is Theorem~\ref{thm Energy reduction in G}.
 It is based on the decomposition~\eqref{eq intro gamma=sum ...}, and it enables to rewrite the kinetic energy per unit surface of three dimensional systems with two-dimensional symmetries (in particular with a density $\rho$ satisfying $\rho(x_1,x_2,x_3)=\rho(0,0,x_3)$), and subject to a constant magnetic field, as an energy functional defined on one-dimensional states. Roughly speaking, Theorem~\ref{thm Energy reduction in G} shows that for a self-adjoint operator $\gamma$ on $L^2(\RR^3)$,  satisfying $0\le\gamma\le1$, there exists a trace class operator $0\le G_{\gamma}$  acting on $L^2(\RR)$ with the same density ($\rho_{\G_\gamma}(x)=\rho_\gamma(0,0,x)$) such that
 \[ \frac{1}{2}\VTr_{2}(\bL_\bA^{3d}\gamma) = \frac{1}{2}\frac{b_3^2}{\av{\bB}^2}\Tr(-\Delta G_\gamma)+\frac{\av{\bB}}{b_3}\Tr(\omega^{2d}(b_3,G_\gamma)).
 \]
 This result would allow to reduce any DFT model of the form 
 $$
 \cE(\gamma)= \frac{1}{2}\VTr_{2}(\bL_\bA^{3d}\gamma)+\cF(\rho_\gamma)
 $$
 to a model posed on 1d density matrices $G$ acting on $L^2(\RR)$, similarly as it is done in~\cite{GLM23} for the reduced Hartree-Fock model.  
 
 \medskip
This article is structured as follows.  In Section~\ref{sect preliminaries}, we recall same basic properties of the Landau operator, magnetic translations and  the  harmonic oscillator. We also introduce a type of Wigner transform that we need and recall the Moyal identity. In Section~\ref{sec:diagonalisation}, we give the main results of decomposition of 2d and 3d operators commuting with 2d magnetic translations and Section~\ref{sec model reduction} is devoted to the reduction of the kinetic energies per unit surface and volume, defined through thermodynamic limits, using the spectral decomposition in Theorem~\ref{thm decomposition of 2d operators for B orthogonal} and Theorem~\ref{thm decomposition of 3d operators for B orthogonal}. Section~\ref{sect proofs} is dedicated to the proofs of the main results. In Appendix~\ref{appendix}, we discussed the behavior of $\omega^{3d}(b,\rho)$ as a function of $b$.

\subsection*{Notation}\label{sec notation}Throughout this paper, we make use of the following notation:
\begin{itemize}
	\item $\bS(L^2(\RR^d))$ stands for the space of bounded self-adjoint operators on $L^2(\RR^d)$; for $d\ge1$. \item $\cS(\RR^d)$ refers to the classical Schwartz space of smooth fast decaying functions of $\RR^d$, $d\ge1$. 
	\item We define the partial Fourier transform in dimension $d$ in the $x_j$-direction, $1\le j\le d$,  as the unitary map denoted by $\cF_{j}:L^2(\RR^d)\to L^2(\RR^d)$ and given by
\begin{equation*}
    \cF_j(f)(x_1,\ldots,x_{j-1},k,x_{j+1},\ldots,x_d) = \frac{1}{\sqrt{2\pi}} \int _{\RR} f(\bx)e^{-\ri kx_j}\,dx_j,\quad \forall f \in \cS(\RR^d).
\end{equation*}
\item For $\bR\in\RR^d$, $d\ge1$, we denote by $\tau_\bR$ the translation operator $\tau_\bR f=f(\cdot-\bR)$, $f\in L^2(\RR^d)$.
\item If $0\le\gamma\in\bS(L^2(\RR^d))$, we say that $\gamma$ is locally of finite trace if $\Tr(\1_Q\gamma\1_Q)<\infty$, for all bounded measurable $Q\subset\RR^d$.\\ 
\item For $L>0$ and $d\ge1$, we write $\Gamma_L^d:=\left[-\frac{L}{2},\frac{L}{2}\right]^d$.
\item For $d\in\set{2,3}$ and $\gamma\in \bS(L^2(\RR^d))$, we denote by $\VTr_2(\gamma)$ the trace per unit surface of $\gamma$ given by, upon existence, 
\[ \VTr_2(\gamma)=\lim_{L\to\infty}\frac{1}{L^2}\Tr(\1_{\Gamma_L^2\times\RR^{d-2}}\,\gamma\,\1_{\Gamma_L^2\times\RR^{d-2}}),
\]
with the convention $\RR^{0}=\set{0}$. Similarly,  we  define the trace per unit volume of $\gamma\in\bS(L^2(\RR^3))$ as 
\[ \VTr_3(\gamma)=\lim_{L\to\infty}\frac{1}{L^3}\Tr(\1_{\Gamma_L^3}\,\gamma\,\1_{\Gamma_L^3}).
\]
\end{itemize}

\subsection*{Acknowledgments} 
A. Maichine acknowledges partial support from the IMU-CDC and Simons Foundation and thanks the Institute f\"ur Mathematik at the University of Oldenburg, as well as Konstantin Pankrashkin, for their support and hospitality.
S. Lahbabi was a Fellow-in-Residence at CY Advanced Studies during the completion of this work. She gratefully acknowledges the institute, the AGM laboratory, and Constanza Rojas-Molina for their support and hospitality. CJGC acknowledges support from the AFOSR through grant FA9550-18-1-0095.

\section{General setting and preliminaries}\label{sect preliminaries}
\subsection{Landau operator}\label{sec:landau}
The kinetic energy operator for a spinless electron gas in the uniform magnetic field $\bB = (b_1,b_2,b_3)$, usually called the Landau operator, is given by 
\begin{equation}\label{Landau Operator}
\bL_\bA^{3d}:=(\bp+\bA)^2=\bra{\bp_1^\bA}^2+ \bra{\bp_2^\bA}^2+\bra{\bp_3^\bA}^2,    
\end{equation}
where $\bp:=-\ri\nabla$, $\bp_j^\bA=-\ri\,\partial_j+a_j(\bx)$, for $j\in\{1,2,3\}$, and the magnetic vector potential $\bA=(a_1,a_2,a_3)$ satisfies ${\rm curl}\, \bA = \bB$. 

In this article, we choose $\bA=(b_2x_3,b_3x_1, b_1x_2)^T$ for convenience. Note that ${\rm curl}\,\bA = \bB$, and ${\rm div}\, \bA = 0$.  Any other choice of magnetic vector potential can be reduced to this one through a gauge transformation. 

With this choice, 
\begin{equation}
\bp_1^\bA :=\bp_1+b_2x_3,\quad \bp_2^\bA:=\bp_2+b_3x_1,\quad\text{and}\quad \bp_3^\bA:=\bp_3+b_1x_2.    
\end{equation}
These magnetic momentum operators satisfy the following commutation relations
\begin{equation}
[\bp_1^ \bA, \bp_2^ \bA]=-\ri b_3, \quad [\bp_2^ \bA, \bp_3^ \bA]= -\ri b_1,\quad\text{and}\quad [\bp_3^ \bA, \bp_1^ \bA]=-\ri b_2.    
\end{equation}
Therefore, they do not commute among themselves, and
none of them commutes with the Landau operator $\bL_\bA^{3d}$. More generally, 
\begin{equation}
\com{\bp_j^ \bA, \bp_k^ \bA}=-\ri \epsilon_{\ell jk} ({\rm curl}\, \bA)_\ell,
\end{equation}
where $\epsilon_{\ell jk}$ is the Levi-Civita tensor. However, one can construct a {\it dual} operator $\bp+\widetilde{\bA}$ that commutes with $\bp+\bA$. A dual gauge $\widetilde{\bA}$ is chosen so that it satisfies 
\begin{equation}
\com{\bp^\bA_j,\bp^{\widetilde{\bA}}_k}=0,\quad \forall 1\leq j,k\leq 3. 
\end{equation}
Solving the above system, one can write $\widetilde{\bA}:=(b_3 x_2, b_1 x_3, b_2 x_1)$, so that 
\begin{equation}
\widetilde{\bp}_1^\bA:=\bp_1+b_3x_2,\quad \widetilde{\bp}_2^\bA:=\bp_2+b_1x_3\quad\text{and}\quad 
\widetilde{\bp}_3^\bA:=\bp_3+b_2x_1.
\end{equation}
It immediately follows that 
\begin{equation}\label{eq Landau commutes with dual momentum}
    \com{\bL_\bA^{3d} ,\widetilde{\bp}^\bA_k}=0,\quad  \forall 1\leq k\leq 3.
\end{equation}
\subsection{Magnetic translations}\label{sec:magnetic-translations}
Although the magnetic field is uniform and therefore translation invariant, the Landau operator (\ref{Landau Operator}) is not, due to the fact that the vector potential $\bA$ is not itself translation invariant. 
Alternatively, $\bL^{3d}_\bA$ commutes with the (self-adjoint) dual momentum operator $\bp+\tilde{\bA}$ as previously mentioned in \eqref{eq Landau commutes with dual momentum}.

The magnetic translations corresponding to $\bB$ and our chosen magnetic potential $\bA$ are defined as the family of operators $\left(\widetilde{\fm}_\bR^\bB\right)_{\bR\in \RR^3}$ acting on $ L^2(\RR^3)$ as follows 
\begin{equation}
\widetilde{\fm}_\bR^\bB := \exp(-\ri \bR\cdot \bp_{\widetilde{\bA}})=\exp(-\ri \bR\cdot (\bp+\widetilde{\bA})). 
\end{equation}
Thanks to the Baker-Campbell-Hausdorff formula (see, for instance, \cite[Theorem 5.1]{Hall-LieGroups}), $\widetilde{\fm}_\bR^\bB$ are explicitly given by
\begin{equation}
\widetilde{\fm}_\bR^\bB = \re^{\frac{\ri}{2}\btheta(\bB,\bR)}\, \re^{-\ri (b_3 R_1 x_2 + b_1 R_2 x_3 + b_2 R_3 x_1)}\,\tau_\bR,\quad 
\end{equation}
where $\btheta(\bB,\bR) =b_3 R_1 R_2 + b_2 R_1 R_3 + b_1 R_2 R_3$, for any $\bR\in\RR^3$. 
For a more compact representation, we set $\fm_\bR^\bB:=\re^{-\frac{\ri}{2}\btheta(\bB,\bR)} \widetilde{\fm}_\bR^\bB $, with 
\begin{equation}\label{eq expression of 3d magnetic translation}
\fm_\bR^\bB = \re^{-\ri (b_3 R_1 x_2 + b_1 R_2 x_3 + b_2 R_3 x_1)}\,\tau_\bR,\quad\forall\bR\in\RR^3. 
\end{equation}
From now on, we refer to  $\left(\fm_\bR^\bB\right)_{\bR\in \RR^3}$ 
as the the three dimensional magnetic translations. It is worth mentioning that
$$\com{\bL^{3d}_\bA,\fm_\bR^\bB}=0,\quad\forall\bR\in\RR^3.$$
The magnetic translations $\left(\fm_\bR^\bB\right)_{\bR\in \RR^3}$ do not form a group; rather, in our given gauge, they satisfy the following multiplication rule
\begin{equation*}
\fm_\bR^\bB\fm_{\widetilde{\bR}}^\bB=e^{i\widetilde{\bR}\cdot\widetilde{\bA}(\bR)} \fm_{\bR+\widetilde{\bR}}^\bB = 
    \re^{\ri (b_3 R_2\widetilde{R}_1+b_1R_3\widetilde{R}_2 +b_2 R_1\widetilde{R}_3} ) \fm_{\bR+\widetilde{\bR}}^\bB.
 \end{equation*}   
 It follows that the magnetic translations do not commute with each other. However, they are unitary operators and their inverse is given by 
\begin{equation*}
\left (\fm_\bR^\bB\right )^*=\left (\fm_\bR^\bB\right )^{-1}=e^{i\bR\cdot\widetilde{A}(\bR)}\fm_{-\bR}^\bB= \re^{\ri (b_3 R_2 R_1+ b_1 R_3 R_2+b_2R_1R_3)}\,\fm_{-\bR}^\bB.
\end{equation*} 
 In this paper, we are mainly concerned with magnetic translations $\fm_\bR^\bB$ in $\RR^2$, i.e. for $\bR=(R_1,R_2,0)$. From now on, we denote by $(\fm_\bR^{\{b_1,b_3\}})_{\bR\in\RR^2}$ the family of unitary operators on $L^2(\RR^3)$ defined by
\begin{equation}\label{eq formula m_R for R in R**2}
	\fm_\bR^{\{b_1,b_3\}} f(\bx)= \re^{-\ri b_3R_1x_2 -\ri b_1R_2x_3}f(x_1-R_1, x_2-R_2, x_3),
\end{equation}
for all $\bR=(R_1,R_2)\in\RR^2$, all $f\in L^2(\RR^3)$, and all $\bx=(x_1,x_2,x_3)\in\RR^3$. If $b_1=0$, $\fm_\bR^{b_3}:=\fm_\bR^{\{0,b_3\}} $ can be also seen as a unitary operator on $L^2(\RR^2)$. The 2d counterpart of $\fm_\bR^{b_3}$ will be also denoted in the same way. 
One has
\begin{equation}
	\fm_\bR^{b_3} f = \re^{-\ri b_3R_1x_2 }\,\tau_\bR f,\qquad \forall f\in L^2(\RR^2),\; \forall\bR\in\RR^2.
\end{equation}

\subsection{Quantum harmonic oscillator}
In order to better explore the spectral properties of the Landau operator, 
we present, here, a review of some basic properties of the quantum harmonic oscillator. 
For $\alpha>0$, we consider the one-dimensional quantum harmonic oscillator,
\begin{equation}
\mathcal{H}_\alpha = -\frac{d^2}{dx^2}+\alpha^2 x^2.
\end{equation}
We recall that $\mathcal{H}_\alpha$ has a self--adjoint realization on $L^2(\RR)$ (see  \cite[Theorem X.28]{ReeSim2}, for example) that has a compact resolvent as a consequence of the Rellich–Kondrachov theorem \cite{LL}. Moreover, its eigenfunctions and eigenvalues are explicitly known \cite{PaulingWilson}, and the Hamiltonian has the spectral decomposition 
\begin{equation}\label{eq spectral decomp of harmonic oscill}
\mathcal{H}_\alpha = \sum_{n\in\NN_0} \epsilon_n^\alpha \,|\varphi_n^\alpha\rangle\langle\varphi_n^\alpha|,
\end{equation}
where
\begin{equation}\label{eq def of Landau spectrals epsilon_n, phi_n}
\epsilon_n^\alpha := (2n+1)\alpha,\quad   \varphi_n^\alpha=\alpha^{1/4} \varphi_n(\sqrt{\alpha}\,\cdot)    
\end{equation}
and $\varphi_n$ denotes the normalized $n$-th Hermite-Gauss function, for each $n\in\NN_0$:
\begin{equation}\label{Hermite-Gauss function}
\varphi_{n}(x)=\frac {1}{\sqrt {2^{n}\,n!}}\left(\frac {1 }{\pi}\right)^{1/4}e^{-\frac{x^2}{2}}H_{n}\left(x\right),\quad x\in \RR, \quad n\in\NN_0,
\end{equation}
where 
\begin{equation}
H_{n}(x)=(-1)^{n}\,e^{x^2}\frac {d^{n}}{dx^{n}}\left(e^{-x^2}\right),\quad x\in \RR.
\end{equation}
is the $n$-th Hermite polynomial. 
\subsection{Anisotropic harmonic oscillator}\label{app:anisotropic-harmonic-oscillator}
In this section we describe the connection between the Landau operator $L_{\bA}^{3d}$  in (\ref{Landau Operator}) and the quantum harmonic oscillator. More specifically, we recall that, if the magnetic field is two-dimensional, the Landau operator admits a fiber decomposition in which each fiber represents a two-dimensional (anisotropic) harmonic oscillator, which allows us to characterize the spectrum of $L_{\bA}^{3d}$ completely. 

Without loss of generality, we can assume that $b_1=0$. In this case,  
\[ \bL^{3d}_\bA=(\bp+\bA)^2=\bra{\bp_1+b_2 x_3}^2+ \bra{\bp_2+b_3 x_1}^2+\bra{\bp_3}^2, 
\]
Now, considering the partial Fourier transform in the second component, $\cF_2$, we can write $\bL^{3d}_\bA$ as a direct integral 
\[
\cF_2 \bL^{3d}_\bA \cF_2^{-1} = \int_{\RR}^\oplus \cH_{2d}[k]\,dk,
\]
where 
\[ \cH_{2d}[k]:=\bra{\bp_1+b_2 x_3}^2+ b_3^2\bra{x_1+\frac{k}{b_3}}^2+\bp_3^2
= \tau_{\left (-\frac{k}{b_3},0\right )}\;\cH_{2d}\;\tau_{\left (\frac{k}{b_3},0\right )},
\]
and 
\begin{equation}\label{eq anisotropic H-2d} \cH_{2d}\equiv \cH_{2d}[0]=\bra{\bp_1+b_2 x_3}^2 + b_3^2 x_1^2+\bp_3^2.
\end{equation}
 Otherwise, if $b_2\neq 0$, $\cH_{2d}$ is an anisotropic harmonic oscillator whose spectral properties will be summarized here. \\
Consider now the unitary map $\mathcal{V}:L^2(\RR^2)\to L^2(\RR^2)$ given by 
\[
\mathcal{V}f(x_1,x_3) = \re^{\ri b_2x_1x_3} f(x_1,x_3),\quad \forall\, f \in L^2(\RR^2).
\]
Then, by a gauge transform, we obtain
\[ 
\widetilde{\cH}_{2d}:=\mathcal{V}\,\cH_{2d} \,\mathcal{V}^* = \bp_1^2 + b_3^2 x_1^2 + (\bp_3-b_2 x_1)^2.
\]
Now, $\widetilde{\cH}_{2d}$ can be decomposed via the partial Fourier transform in $x_3$ as 
\[ \mathcal{F}_3\,\widetilde{\cH}_{2d} \,\mathcal{F}_3^{-1} = \int_\RR^\oplus h(k) \rd k,
\]
where 
\begin{align*}
h(k) &:= \bp_1^2 + b_3^2 x_1^2 + (k-b_2 x_1)^2
= -\partial_{x_1}^2 +|\bB|^2\left(x_1+c_k\right )^2 +a_k,
\end{align*}
$c_k = b_2k/|\bB|$, and $a_k = \left(1-b_2^2/|\bB|^2\right) k^2$. For each $k\in \RR$, $h(k)$ is a one-dimensional quantum harmonic oscillator with frequency $|\bB|$, centered at $c_k$, with an energy shift $a_k$, and therefore its eigenvalues are 
\begin{equation}
    \epsilon^{|\bB|}_n(k) = |\bB|(2n+1) + \left(1-\frac{b_2^2}{|\bB|^2}\right) k^2,\quad n\in \NN_0, 
\end{equation}
with corresponding eigenfunctions 
\begin{equation}
    \varphi_n^{|\bB|}(x;k) = |\bB|^{1/4} \varphi_n\left ( |\bB|^{1/2}(x - c_k)\right ),\quad x\in \RR,\ n\in \NN_0.
\end{equation}
The spectrum of $h(k)$ is, therefore, 
\[
\sigma\left (h(k)\right )=\left \{  |\bB|(2n+1)+ \left(1-\frac{b_2^2}{|\bB|^2}\right) k^2 :\  n\in\NN_0\right \}.
\]
If $b_3\neq 0$, the spectrum of $\cH_{2d}$ becomes
\[
\sigma\left ({\cH}_{2d}\right )=\sigma\left (\widetilde{\cH}_{2d}\right ) = \bigcup_{k\in \RR} \sigma\left (h(k)\right )=\left [|\bB|,+\infty\right ).
\]
Otherwise, if $b_2\ne 0$ and $b_1=b_3=0$, the band levels $\epsilon_n^{b_2}$ {\it flatten out}, and the spectrum of the Landau operator becomes discrete and equal to that of the one-dimensional quantum Harmonic oscillator with $\alpha=|b_2|$. 

\subsection{Wigner-type transform}\label{sec Wigner transform}
For $\alpha>0$ and $f,g\in\cS(\RR)$, we define the Fourier-Wigner transform as 
\begin{equation}\label{eq def of Wigner transform}
\cW_\alpha^{2d} (f,g) (x_1,x_2)=\frac{1}{\sqrt{2\pi}}\int_\RR f\left(x_1-\frac{k}{\alpha}\right)\overline{g(k)}\re^{-\ri k x_2} \rd k.
\end{equation}
\begin{remark}
We have adopted here the definition in \cite{GLM23}, which slightly  differs from the definition in \cite[Chapter 2]{Wong1998}. 
\end{remark}
We summarize some of the properties of this transform. More details can be found in \cite{GLM23}. For $f,g\in \cS(\RR)$, we have $\cW_\alpha^{2d} (f,g)\in\cS(\RR^2)$, and $\cW_\alpha^{2d} $ extends to an isometry from $L^2(\RR)\times L^2(\RR)$ to $L^2(\RR^2)$. This follows from the Moyal identity: for all $f_j,g_j\in \cS(\RR)$, $j\in\{1,2\}$
\begin{equation}\label{eq. Moyal identity}
\langle \cW_\alpha^{2d}  (f_1,g_1),\cW_\alpha^{2d}  (f_2,g_2) \rangle_{L^2(\RR^2)} = \langle f_1, f_2\rangle_{L^2(\RR)}\, \langle g_1, g_2\rangle_{L^2(\RR)}. 
\end{equation}
For $f\in\cS(\RR^2)$ and $g\in\cS(\RR)$, we similarly define $\cW_\alpha^{3d}$ as
\begin{equation}\label{eq. def of 3d Wigner}
\cW_\alpha^{3d} (f,g) (x_1,x_2,x_3)=\frac{1}{\sqrt{2\pi}}\int_\RR f\left(x_1-\frac{k}{\alpha},x_3\right)\overline{g(k)}\re^{-\ri k x_2} \rd k.
\end{equation} 
We notice that $\cW_\alpha^{3d} (f,g)(x_1,x_2,x_3)=\cW_\alpha^{2d}(f(\cdot,x_3),g)(x_1,x_2)$ and that $\cW_\alpha^{3d}$ defines an isometry between $L^2(\RR^2)\times L^2(\RR)$ and $L^2(\RR^3)$.
\section{Statement of the main results}\label{sec main results}

\subsection{Characterization of operators commuting with $\fm_\bR^{\bB}$}\label{sec:diagonalisation}



The first main result of this work is the characterization of operators commuting with the magnetic translations $\{ \fm_\bR^{b_3}\}_{\bR\in\RR^2}$ in 2d and $\{ \fm_\bR^{\set{b_1,b_3}}\}_{\bR\in\RR^2}$ in 3d. Theorems~\ref{thm decomposition of 2d operators for B orthogonal} and~\ref{thm decomposition of 3d operators for B orthogonal} below can be viewed as natural analogues of the classical result that operators commuting with translations are the multiplication operators in the Fourier space. Our framework is motivated by applications to the reduction of DFT models, in the presence of magnetic fields for three-dimensional electronic systems with two-dimensional symmetry; see for instance the  previous work~\cite{GLM23}.

\begin{theorem}[Two-dimensional Case]\label{thm decomposition of 2d operators for B orthogonal}
	Let $b\neq 0$. Let $\eta$ be a non-negative locally trace class self-adjoint operator on $L^2(\RR^2)$ satisfying
	$$
	\com{\eta, \fm_\bR^{b}}=0, \quad \forall \bR\in \RR^2.
	$$
	Then, there exists an orthonormal basis $\{\psi_n\}_{n\in\NN}$ of $L^2(\RR)$ and a sequence of nonnegative summable real numbers $\{\lambda_n\}_{n\in\NN}$ such that 
	$$
	\eta=\sum_{n\in \NN} \lambda_n\bK_{\psi_n}^{2d},
	$$
	where $\bK_{\psi_n}^{2d}$ is the orthogonal projector onto $E_{\psi_n}^{2d} = \set{\cW_{b}^{2d}(\psi_n,g),\; g\in L^2(\RR)}$.
Moreover, one has
\begin{equation}
\VTr_{2}(\eta)=\frac{b}{2\pi}\sum_{n\in\NN} \lambda_n.
\end{equation}
\end{theorem}
\begin{remark}
Theorem~\ref{thm decomposition of 2d operators for B orthogonal} generalizes \cite[Proposition 2.5]{GLM23} to the case where the operator $\eta$ is not required to commute with the Landau operator. 
\end{remark}
The following theorem is the three-dimensional counterpart of the previous result.

\begin{theorem}[Three-dimensional Case]\label{thm decomposition of 3d operators for B orthogonal}
	Let $b_1\in \RR$ and $b_3\neq 0$. Let $\gamma$ be a nonnegative self-adjoint operator on $L^2(\RR^3)$ satisfying
	$$
	\com{\gamma, \fm_\bR^{\set{b_1,b_3}}}=0, \quad \forall \bR\in \RR^2.
	$$
	Assume that $\gamma$ has a finite trace per unit surface. Then, there exists an orthonormal basis $(\psi_n)_n$ of $L^2(\RR^2)$ and a sequence of nonnegative summable real numbers $(\lambda_n)_n$ such that 
	$$
	\gamma=\sum_{n\in \NN} \lambda_n\bK_{\psi_n}^{3d},
	$$
	where $\bK_{\psi_n}^{3d}$ is the orthogonal projector onto $\set{\cW_{b_3}^{3d}(\psi_n,g),\; g\in L^2(\RR)}$.
Moreover, one has
\begin{equation}\label{eq rho_gamma(x_3)=...}
\rho_\gamma(x_3)= \frac{b_3}{2\pi}\sum_n \lambda_n \int_\RR |\psi_n(x_1,x_3)|^2 \, \rd x_1,\quad \text{for almost all } x_3\in \RR,
\end{equation}
$\rho_\gamma\in L^1(\RR)$, and
\begin{equation}
\VTr_3(\gamma)=\int_{\RR} \rho_\gamma = \frac{b_3}{2\pi}\sum_n \lambda_n.
\end{equation}
\end{theorem}
The proofs of these two theorems are presented in Section~\ref{sect proofs-characterization}.\\
\medskip

Unlike the case of invariance by ordinary translations where the Fourier multipliers commute, the invariant operators by magnetic translations do not commute in general. Actually, one has   
\begin{proposition}\label{prop C*-algebra}
	Let $\psi,\widetilde{\psi}\in L^2(\RR)$ such that $\|\psi\|=\|\widetilde{\psi}\|=1$. Then,
	\begin{equation}\label{eq Tr(K_psi K_tildepsi)=...}
		\VTr_2 (\bK_{\psi}^{2d}\bK_{\widetilde{\psi}}^{2d})=\frac{b_3}{2\pi}\, |\langle\psi,\widetilde{\psi}\rangle|^2  .
	\end{equation}
	and
	\begin{equation}\label{eq com(K_psi,K_tildepsi)=0 iff ...}
		\com{\bK_{\psi}^{2d},\bK_{\widetilde{\psi}}^{2d}}=0 \quad \Longleftrightarrow \quad \psi \perp \widetilde{\psi}\; \text{or}\;\, \psi=\pm\widetilde{\psi}.
	\end{equation}
\end{proposition}
The proof of the proposition is given in Section~\ref{sect proofs-characterization}.
\medskip

The projectors $\bK^{3d}_\psi$ satisfy properties similar to the ones of  $\bK^{2d}_\psi$ given in Proposition~\ref{prop C*-algebra}. 

\subsection{Reduction of the kinetic energy functional}\label{sec model reduction}
The spectral decompositions presented in Theorem~\ref{thm decomposition of 2d operators for B orthogonal} and Theorem~\ref{thm decomposition of 3d operators for B orthogonal} allow us to reduce the kinetic energies of non-interacting electron gases in both two- and three-dimensional systems in the presence of a magnetic field. 
\subsubsection{Two-dimensional homogenous electron gas}
We consider a two-dimen\-sional homogeneous electron gas with constant  density $\rho>0$. Let $b>0$ be the strength of a magnetic field applied in the $x_3$-direction, orthogonal to the electron gas. We aim to calculate the kinetic energy density $\omega^{2d}(b,\rho)$ of $\rho$ under the action of the external field $\bB=(0,0,b)$. In this paper, we define the kinetic energy density $\omega^{2d}(b,\rho)$ as
\begin{align}\label{eq def F(b,rho) as 2d kin energy}
\omega^{2d}(b,\rho):=\inf\bigg \{  \frac{1}{2}\VTr_2(\bL_\bA^{2d}\eta) : & \ \eta\in\bS(L^2(\RR^2)),\; 0\le \eta\le1,\; \rho_\eta=\rho,\; \\  \nonumber & \text{and}\; [\eta,\fm_\bR^b]=0 \bigg \}. 
\end{align}

As the two-dimensional Landau operator 
\begin{equation}
	\bL^{2d}_\bA := \bp_1^2+(\bp_2+bx_1)^2
\end{equation} 
commutes with $\fm_\bR^b$, and the energy functional $\eta \mapsto \frac{1}{2}\Tr(\bL_\bA^{2d}\eta)$ is linear,

Using suitable boundary conditions and classical techniques, see~\cite{Iwatsuka, GLM21}, one can show that the expressions~\eqref{eq:omega} and~\eqref{eq def F(b,rho) as 2d kin energy} are actually equal. The following proposition gives an explicit expression of it. 
  
\begin{proposition}\label{prop F(b,rho)=...}
	The ground state kinetic energy $\omega^{2d}(b,\rho)$ has the explicit expression
\begin{equation}
\omega^{2d}(b,\rho)=\pi \rho^2+\frac{b^2}{4\pi}\set{\frac{2\pi \rho}{b}}\bra{1-\left\{\frac{2\pi \rho}{b}\right\}},
\end{equation}
where $\{x\}:=x-\lfloor x \rfloor$ refers to the fractional part of $x\in\RR$.
\end{proposition}
The functional $\omega^{2d}$ plays the role of the kinetic energy density in two-dimen\-sional Thomas-Fermi-like functional energies in the presence of a uniform magnetic field, see~\cite{LiebSolYng_95}.
Next, we provide some elementary properties of the functional $\omega^{2d}$. In particular, letting $b\to0$, we retrieve the non-magnetic Thomas-Fermi kinetic energy $\pi \rho^2$ in dimension two. 
\begin{corollary}
Let $\omega^{2d}(0^+,\rho)=\displaystyle\lim_{b\to0}\omega^{2d}(b,\rho)$. Then,
\begin{enumerate}
\item For all $m\in\NN$,
\[ \inf_{b>0} \omega^{2d}(b,\rho)=\omega^{2d}(0^+,\rho)=\pi\rho^2=\omega^{2d}\left (\frac{2\pi\rho}{m},\rho\right ),
\]
\item
\[ \pi\rho^2\le \omega^{2d}(b,\rho)\le \pi\rho^2 +\frac{b^2}{16\pi}.
\]
\item $x\mapsto \omega^{2d}(b,x)$ is increasing and piecewise linear.
\item $x\mapsto \omega^{2d}(b,x)-\pi x^2$ is $\left (\frac{b}{2\pi}\right )$--periodic.  
\end{enumerate}
\end{corollary}
The proof of Proposition~\ref{prop F(b,rho)=...} can be found in Section~\ref{sect proofs-2} and the proof of the corollary, as well as further properties of the functional $(b,\rho)\mapsto \omega^{2d}(b,\rho)$, can be found in  \cite[Section 3.2]{GLM23}. 

\subsubsection{Three-dimensional homogeneous electron gas}
We consider now a three-dimensional homogeneous electron gas with constant  density $\rho>0$.
As the system is rotationally invariant, we can assume, without loss of generality, that $\bB$ is of the form $\bB=(0,0,b)$, with $b>0$. Since $\bL_\bA^{3d}$ commutes with the magnetic translations $(\fm_\bR^{b})_{\bR\in\RR^3}$, similarly as in the 2d case, 
the three-dimensional  kinetic energy density of a homogeneous electron gas under the magnetic field $\bB$ can be written as 
\begin{align}
\omega^{3d}(b,\rho):=\inf\bigg \{ \frac{1}{2}\VTr_{3}(\bL_\bA^{3d} \gamma): &\ \gamma\in\bS(L^2(\RR^3)), 0\le\gamma\le1, \\ \nonumber &  [\gamma,\fm_\bR^b]=0, \forall \bR\in\RR^3\;\text{and}\; \VTr_{3}(\gamma)=\rho \bigg \}.
\end{align} 
In the next result, we give an explicit formula for $\omega^{3d}(b,\rho)$. 
\begin{proposition}\label{prop omega_B(rho)=...}
Let $\rho>0$. Then,
\begin{align}\label{eq omega_B(rho)=}
\omega^{3d}(b,\rho)=\frac{\delta\rho}{6}+\frac{b^2}{6\pi^2}\sum_n \epsilon_n^{b}\left( \delta -\epsilon_n^{b}\right)_+^{1/2},
\end{align}
where $\epsilon_n^{b} = b (2n+1)$, $n\in\NN_0$, are the Landau levels, introduced in Section~\ref{app:anisotropic-harmonic-oscillator} and the Fermi level $\delta>0$ is the unique solution to
\begin{equation}\label{eq rho function of delta}
\sum_n \left( \delta -\epsilon_n^{b}\right)_+^{1/2} = \frac{2\pi^2 \rho}{b}. 
\end{equation}
\end{proposition}
The proof of Proposition~\ref{prop omega_B(rho)=...} can be read in Section~\ref{sect proofs-3}. 

\begin{remark}
Note that 
\[
g(\delta) := \sum_n \left( \delta -\epsilon_n^{b}\right)_+^{1/2}
\]
is a strictly increasing coercive function of $\delta$, which explains why the solution exists and is unique.
\end{remark}
The functional $\omega^{3d}(b,\rho)$ appears in the magnetic Thomas-Fermi energy functional~\cite{LSY94-1,LSY94-2,MadSor20} instead of the $C_{\rm TF} \rho^{5/3}$ in the non magnetic case. We recover the limit when $b\to 0$ in the following proposition. 

\begin{proposition}\label{prop:asymptot-w3d}
	Let $\rho\geq 0$. Then 
	$$
	\lim_{b\to 0} \omega^{3d}(b,\rho)= \frac{(3\pi^2)^{2/3}}{3}\rho^{5/3}.
	$$
	Moreover, for $b> \bra{2\pi^4 \rho^2}^{1/3}$, one has
	$$
	\omega^{3d}(b,\rho)=\bra{2b^2+b}\frac{\rho}{6}+ \bra{\frac{2\pi}{b}}^{2}\frac{\rho^3}{6}. 
	$$
\end{proposition}

The proof of Proposition~\ref{prop:asymptot-w3d} is detailed in the appendix. 

\begin{remark}
	The constant we recover here is different from the Thomas-Fermi constant $C_{\rm TF}= \frac{3}{10}(3\pi^2)^{2/3}$. 
\end{remark}
\subsubsection{Three-dimensional electronic system with 2d symmetry} 
We now consider a three-dimensional electronic system with 2d symmetries (in particular $\rho_\gamma(x_1,x_2, x_3)=\rho(x_3)$) subject to a constant magnetic field. 
We may assume without loss of generality that the magnetic field is of the form $(0,b_2,b_3)$, with $b_3>0$.  The Landau operator is then
$$
\bL_\bA^{3d}=\bra{\bp_1+b_2 x_3}^2+ \bra{\bp_2+b_3 x_1}^2+\bp_3^2,
$$
and the set of admissible states is 
\begin{equation}
\cK:=\set{ \gamma\in\bS(L^2(\RR^3)): 0\le\gamma\le1,\; [\gamma,\fm_\bR^{b_3}]=0,\,\forall\bR\in\RR^2;\text{and}\;\VTr_2(\gamma)<\ii  }. 
\end{equation} 
We will show that the 3d problem is equivalent to a 1d problem, where the set of admissible state is 
\begin{equation}
\cG:=\set{ G\in \bS(L^2(\RR)) : G\ge0\quad\text{and}\quad \Tr(G)<\ii}. 
\end{equation}
Our main result then reads
\begin{theorem}\label{thm Energy reduction in G}
Let $0\le\rho\in L^1(\RR)$. Then, 
\begin{equation}\label{eq reudction: inf-in-gamma vs inf-in-G_b2}
	\inf_{\gamma \in \cK \atop \rho_\gamma = \rho} \set{ \frac{1}{2}\VTr_2(\bL_\bA^{3d}\gamma)}=\inf_{ G\in \cG \atop \rho_G = \rho} \set{ \frac{1}{2}\frac{b_3^2}{\av{\bB}^2}\Tr(-\Delta G)+\frac{\av{\bB}}{b_3}\Tr(\omega^{2d}(b_3,G)).}.   
\end{equation}
In particular, if $b_2=0$, then
\begin{equation}\label{eq reudction: inf-in-gamma vs inf-in-G}
	\inf_{\gamma \in \cK \atop \rho_\gamma = \rho} \set{ \frac{1}{2}\VTr_2(\bL_\bA^{3d}\gamma)}=\inf_{ G\in \cG \atop \rho_G = \rho} \set{ \frac{1}{2}\Tr(-\Delta G)+\Tr(\omega^{2d}(b_3,G))}.   
\end{equation}
\end{theorem}
This  result can be seen as a generalization of~\cite[Theorem 3.1]{GLM23} to the case where the magnetic field is not orthogonal to the material and where the admissible states do not necessarily need to commute with the Landau operator. The proof of Theorem~\ref{eq reudction: inf-in-gamma vs inf-in-G_b2} is detailed in Section~\ref{sect proofs-4}. 


\section{Proofs of the main results}\label{sect proofs}

\subsection{Proofs of Theorems~\ref{thm decomposition of 2d operators for B orthogonal} and~\ref{thm decomposition of 3d operators for B orthogonal} (Characterization of operators commuting with $\fm_\bR^{\bB}$)}
\label{sect proofs-characterization}
\begin{proof}[Proof of Theorem~\ref{thm decomposition of 2d operators for B orthogonal}]
	Let $\eta\in\bS(L^2(\RR^2))$  such that $[\eta, \fm_\bR^{b}]=0$, for all $\bR\in\RR^2$. In particular, for $\bR=(0,R_2)$, $\fm_\bR^{b}=\tau_{(0,R_2)}$ and  $\eta$ commutes with translations in the $x_2$-direction. Hence, $\eta$ admits a direct integral decomposition into fibers (see for example \cite[Theorem 4.4.7]{Bratteli-v1}). In this case, the fiber decomposition is given by the Fourier transform in the $x_2$-direction, and  there exist $(\eta_k)_{k\in\RR}\subset \bS(L^2(\RR))$  such that
\begin{equation}\label{eq: Bloch-Floquet for eta in x2}
	\cF_2 \eta \cF^{-1}_2= \int_\RR^{\oplus} \eta_k.  
\end{equation}	
	In addition, 
	 it is easy to check that 
\begin{equation}\label{eq: F m(r,0) F-1 = U}
\cF_2  \fm_{(r,0)}^{b} \cF_2^{-1}  = \tau_{(r,br)},\quad \forall\, r\in \RR,\ \forall\, b\in \RR. 
\end{equation}
	Recall that $\tau$ refers to the translation operators introduced in Section~\ref{sec notation}. Now, the relation $[\eta,\fm_{(r,0)}^{b}]=0$, together with \eqref{eq: Bloch-Floquet for eta in x2} and \eqref{eq: F m(r,0) F-1 = U}, 
implies that 
\[
\left [ \int_{\RR}^{\oplus} \eta_k, \tau_{(r,br)}\right ] = 0.
\]
If we consider functions $\phi \in \cS(\RR^2)$ of the form $\phi(x,k)=f(x)g(k)$, we obtain that
	$$
\bra{\eta_k \tau_{r}f}(x) g(k-b r)   = \tau_r \bra{\eta_{k-b r}f}(x)g(k-b r);
$$	  
in other words, 
\[
\eta_k=\tau_r\,\eta_{k-b r}\,\tau_{-r}.	
\]
Taking $r=k/b$, we obtain a characterization of the fibers of $\eta$, in terms of the zero-fiber $\eta_0$
\begin{equation}\label{eqn:fiber identification}
	\eta_k =\tau_{k/b}\, \eta_0\, \tau_{-k/b},\quad \forall\, k \in \RR.    
\end{equation}
In addition, since $\eta$ is locally trace class, it follows that $\eta_k$ is a trace class operator, for all $k\in \RR$. Therefore, $\eta$ admits the following integral kernel (see~\cite[Section 3]{Bensiali:2024})
	$$
	\eta(\bx,\by)=\frac{1}{2\pi}\int_\RR e^{-ik(x_2-y_2)}\eta_k(x_1, y_1)dk,
	$$
where, for every $k\in \RR$, $\eta_k\in L^2(\RR^2)$ also denotes the integral kernel of $\eta_k$.  
    In particular, we can associate a density $\rho_\eta: \RR\to \RR$ to $\eta$, which is given by 
	$$
	\rho_\eta(x_1)=\frac{1}{2\pi}\int_\RR \rho_{\eta_k}(x_1)dk= \frac{1}{2\pi} \int_\RR \rho_{\eta_0}\bra{x_1-\frac{k}{b}}dk= \frac{b}{2\pi} \int_\RR \rho_{\eta_0}<\ii. 
	$$
Notice that, since $\eta$ has local finite trace and it commutes with $\{\fm_\bR^{b}\}_\bR$, then its density is constant $\rho_\gamma(\bx)=\rho_\gamma({\bf 0} )$.
Now, writing the spectral decomposition of $\eta_0$ as 
	$$
	\eta_0=\sum_{n}\lambda_n \av{\psi_n\rangle\langle\psi_n}, 
	$$
with $\lambda_n\in \RR^+$ and $\{\psi_n\}_{n\in\NN}\subset  L^2(\RR)$ is an orthonormal basis, we obtain
	$$
	\eta_k=\sum_{n}\lambda_n \av{\psi_n(\cdot -k/b)\rangle\langle\psi_n(\cdot -k/b)}.
	$$
By \eqref{eq: Bloch-Floquet for eta in x2}, it follows  that for each $n\in\NN$ and any $g\in L^2(\RR)$, $\cW_{b}^{2d}(\psi_n,g)$ is an eigenfunction of $\eta$ with corresponding eigenvalue $\lambda_n$. To see this, first note that we can write 
\[
\cW_{b}^{2d}(f,g) = \cF_2^{-1} \left ((\tau_{\cdot /b}f) \overline{g(\cdot)}\right ). 
\]
Therefore, using (\ref{eqn:fiber identification}), we get for al $(x_1,k)\in \RR^2$
\begin{align}
        \left (\eta \cW_{b}^{2d}(\psi_n,g) \right )_k &=  \eta_k (\tau_{k /b}\psi_n )\overline{ g(k)} = \tau_{k/b}\eta_0   \psi_n \overline {g(k)} \\ \nonumber 
        &=   \lambda_n    \tau_{k/b}\psi_n    \overline {g(k)} = \lambda_n \cF_2 \cW_{b}^{2d}(\psi_n,g)(\cdot,k).
\end{align}
It follows that 
\begin{equation}
    \eta \cW_{b}^{2d}(\psi_n,g) = \lambda_n \cW_{b}^{2d}(\psi_n,g),
\end{equation}
and the claim is proved.
As a consequence, the family of functions 
\[
E= \bigcup_n \left \{ \cW^{2d}_{b}(\psi_n,g):\ g\in L^2(\RR)\right \}
\]
satisfies  $\overline{\text{span } E}= L^2(\RR^2) $. 
Furthermore, the Moyal identity \eqref{eq. Moyal identity} also guarantees that $\left\{\cW_{b}^{2d}(\psi_n,\psi_m): n,m\in\NN\right\}$ forms a complete orthonormal family in $L^2(\RR^2)$. Hence, setting 
\[\bK_{\psi_n}^{2d}=\sum_m |\cW_{b}^{2d}(\psi_n,\psi_m)\rangle\langle \cW_{b}^{2d}(\psi_n,\psi_m)|,
\]
we see that $\bK_{\psi_n}^{2d}$ is the spectral projector onto
\[
E_{\psi_n}^{2d}=\left \{ \cW^{2d}_{b}(\psi_n,g):\ g\in L^2(\RR)\right \} \subseteq  \ker(\eta-\lambda_n).
\]
Recalling that each $\lambda_n\in \sigma(\eta_0)$ has finite multiplicity as an eigenvalue of $\eta_0$,  we can see that
\[
\ker(\eta-\lambda_n) = \bigoplus_{j:\ \psi_j \in \ker(\eta_0-\lambda_n)} E_{\psi_j},
\]
and we obtain the  decomposition of $\eta$  
\[ \eta=\sum_{n} \lambda_n \bK_{\psi_n}^{2d}.
\]  
Moreover, taking into account the fact that, for any $f,g\in L^2(\RR)$,
\[ \cW_{b}^{2d}(f,g)(\bx)= \langle g, \Phi_{f,\bx}\rangle_{L^2(\RR)},
\]
where
\[
\quad\Phi_{f,\bx}(k)=\frac{1}{\sqrt{2\pi}}e^{-ikx_2}{f\bra{x_1-k/b}}, 
\]
we obtain
\begin{align*}
\rho_\eta(\bx)&=\sum_{n,m} \lambda_n |\cW_{b}^{2d}(\psi_n,\psi_m)(\bx)|^2 = \sum_{n,m} \lambda_n |\langle \psi_m , \Phi_{\psi_n,\bx}  \rangle|^2\\
&= \sum_n \lambda_n \|  \Phi_{\psi_n,\bx} \|^2_{L^2(\RR)} =  \sum_n \frac{\lambda_n}{2\pi} \int_\RR \left|\psi_n\left(x_1-\frac{k}{b}\right)\right|^2 \rd k= \frac{b}{2\pi} \sum_n \lambda_n,
\end{align*}
which concludes the proof of Theorem \ref{thm decomposition of 2d operators for B orthogonal}.

\end{proof}

\begin{proof}[Proof of Theorem~\ref{thm decomposition of 3d operators for B orthogonal}] 
If $b_1=0$, the proof follows the same lines as in the 2d case. Let us give a quick sketch. Let $\gamma\in\bS(L^2(\RR^3))$ such that $[\gamma,\fm_\bR^{b_3}]=0$, for all $\bR\in\RR^2$. Then, $\gamma$ commutes with the partial translations in $x_2$-direction. Hence, there exists $(\gamma_k)_{k\in\RR}\subset\bS(L^2(\RR^2))$ such that $\cF_2^{-1}\gamma\cF_2=\int_\RR^\oplus \gamma_k \rd k$.   
Moreover, $$
	\gamma_k =\tau_{(k/b_3,0)}\, \gamma_0\, \tau_{(-k/b_3,0)}.
	$$
Since $$\frac{b_3}{2\pi}\int_{\RR^2} \rho_{\gamma_0}(x_1,x_3) \rd x_1\rd x_3=\VTr(\gamma)=\int_\RR \rho_\gamma(x_3) \rd x_3 <\ii,$$ 
$\gamma_0$ is trace class. Therefore, if $\gamma_0=\sum_n \lambda_n |\psi_n\rangle\langle\psi_n|$ is the spectral decomposition of $\gamma_0$, we claim, similarly to the 2d setting, that $\gamma=\sum_n \lambda_n \bK_{\psi_n}^{3d}$. 
 Finally, one has
\[ \VTr_{2}(\gamma)=\frac{b_3}{2\pi}\int \rho_{\gamma_0} (x_2,x_3)= \frac{b_3}{2\pi}\sum_n \lambda_n.
\]	

\medskip
Assume now that $b_1\neq 0$. 
Let  $\Lambda_{b_1}:L^2(\RR^3)\to L^2(\RR^3)$ be the unitary operator defined by 
\begin{equation}
	\Lambda_{b_1} f (x_1,x_2,x_3) := f(x_1,x_2,x_3+b_1 x_2).
\end{equation}
Then $\Lambda_{b_1}^{-1}=\Lambda_{b_1}^*=\Lambda_{-b_1}$. If we denote by $T_{b_1}=\Lambda_{b_1}\cF_3$, we can see $\fm_\bR^{\{b_1,b_3\}}$ as a unitary transformation of $\fm_\bR^{b_3}=\fm_\bR^{\{0,b_3\}}$ for all $\bR\in\RR^2$. 
\begin{lemma}
	Let $\bR\in\RR^2$. Then,
	\[ T_{b_1} \fm_\bR^{\{b_1,b_3\}} T_{b_1}^*  = \fm_\bR^{\{0,b_3\}}.
	\]
\end{lemma}
\begin{proof}
	The proof follows from direct a computation. Let $f\in \cS(\RR^3)$ and  $\bR\in\RR^2$ and let us show that 
	\[ 
	T_{b_1} \fm_\bR^{\{b_1,b_3\}} f = \fm_\bR^{\{0,b_3\}}T_{b_1}f.
	\]
	Let
	\begin{equation*}
		g(x_1,x_2,x_3) = \fm_{\bR}^{\{b_1,b_3\}} f(x_1,x_2,x_3) = e^{-i(b_3R_1x_2 + b_1R_2x_3)} f(x_1-R_1,x_2-R_2,x_3). 
	\end{equation*}
	One has,
	\begin{align*}
		\cF_3 g\, (x_1,x_2,k_3) 
		= e^{-ib_3R_1x_2}\cF_3f(x_1-R_1,x_2-R_2,k_3+b_1R_2).
	\end{align*}
	It follows that 
	\begin{align*}
		T_{b_1} \fm_\bR^{\{b_1,b_3\}} f(x_1,x_2,k_3) &= \Lambda_{b_1} \cF_3 g\, (x_1,x_2,k_3) \\
		&= e^{-ib_3R_1x_2}\cF_3f(x_1-R_1,x_2-R_2,k_3+b_1x_2).
	\end{align*}
	Then,
	\begin{align*}
		\fm_\bR^{\{0,b_3\}}T_{b_1}f(x_1,x_2,k_3)&  = \fm_\bR^{\{0,b_3\}}\Lambda_{b_1} \cF_3 f(x_1,x_2,k_3) \\
		&= \fm_\bR^{\{0,b_3\}}\cF_3 f(x_1,x_2,k_3+b_1x_2) \\ 
		& = e^{-ib_3R_1x_2}\cF_3f(x_1-R_1,x_2-R_2,k_3+b_1x_2)\\
		&=T_{b_1} \fm_\bR^{\{b_1,b_3\}} f (x_1,x_2,k_3),\quad \forall\, f\in \cS(\RR^3),
	\end{align*}
	which concludes the proof.
\end{proof}
As a consequence, as $\gamma\in\bS(L^2(\RR^3))$  commutes with $\fm_\bR^{\{b_1,b_3\}}$, $T_{b_1}\gamma T_{b_1}^*$ commutes with $\fm_\bR^{b_3}$. We can then apply the result proved for $b_1=0$ and write $T_{b_1}\gamma T_{b_1}^*$ as
$$
T_{b_1}\gamma T_{b_1}^*=\sum_{n\in\NN}\lambda_n \bK_{\psi_n}^{3d}, 
$$
with $\bK_{\psi_n}^{3d}$ being the orthogonal projector onto $\set{\cW_{b_3}^{3d}(\psi_n,g),\; g\in L^2(\RR)}$ and $(\psi_n)$ being a suitable orthonormal basis. 
It follows that 
$$
\gamma =\sum_{n\in\NN}\lambda_n T_{b_1}^*\bK_{\psi_n}^{3d}T_{b_1}= \sum_{n\in\NN}\lambda_n \widetilde{\bK}_{\psi_n}^{3d}.
$$
Here $\widetilde{\bK}_{\psi_n}^{3d}$ refers to the orthogonal projector onto $\set{\cW_{b_3}^{3d}(\widetilde{\psi}_n,g),\; g\in L^2(\RR)}$, where $\widetilde{\psi}_n(x_1,p):=\cF_2(\psi_n)(x_1-b_1p/b_3, p)$.  
\end{proof}
\subsubsection{Proof of Proposition~\ref{prop C*-algebra}}
Let $\{g_k\}_{k\in\NN}$ be an orthonormal basis of $L^2(\RR)$. Using the Moyal identity \eqref{eq. Moyal identity}, we have 
\begin{align*}
	\bK_{\psi}^{2d}\bK_{\widetilde{\psi}}^{2d}&=\sum_k \sum_\ell \langle\cW_{b}^{2d}(\psi,g_k),\cW_{b}^{2d}(\widetilde{\psi},g_{\ell})\rangle\; \left|\cW_{b}^{2d}(\psi,g_k)\rangle \langle\cW_{b}^{2d}(\widetilde{\psi},g_\ell)\right|\\
	&= \sum_k \sum_\ell \langle\psi,\widetilde{\psi}\rangle \langle g_k,g_\ell \rangle\; |\cW_{b}^{2d}(\psi,g_k)\rangle \langle\cW_{b}^{2d}(\widetilde{\psi},g_\ell)|\\
	&= \langle\psi,\widetilde{\psi}\rangle  \sum_k  |\cW_{b}^{2d}(\psi,g_k)\rangle  \langle\cW_{b}^{2d}(\widetilde{\psi},g_k)|. 
\end{align*}
We recall that 
$$
\cW_{b}^{2d}(f,g)(\bx)= \langle g, \Phi_{f,\bx}\rangle,\quad \text{where} \quad\Phi_{f,\bx}(k)=\frac{1}{\sqrt{2\pi}}e^{-ikx_2}{f\bra{x_1-k/b}}. 
$$
Therefore,
\begin{align*}	\trv_2\bra{\bK_{\psi}^{2d}\bK_{\widetilde{\psi}}^{2d}}=\rho_{\bK_{\psi}^{2d}\bK_{\widetilde{\psi}}^{2d}}(\bx)&= \langle\psi,\widetilde{\psi}\rangle \sum_k  \cW_{b}^{2d}(\psi,g_k)(\bx)\overline{\cW_{b}^{2d}(\widetilde{\psi},g_k)(\bx)}\\
	&= \langle\psi,\widetilde{\psi}\rangle \sum_k  {\langle g_k, \Phi_{\psi,\bx}   \rangle }\langle  \Phi_{\widetilde{\psi},\bx} , g_k \rangle\\
	&= \langle\psi,\widetilde{\psi}\rangle \langle   \Phi_{\widetilde{\psi},\bx} , \Phi_{\psi,\bx} \rangle \\
	&= \langle\psi,\widetilde{\psi}\rangle \frac{1}{2\pi } \int_\RR \overline{\widetilde{\psi}\bra{x_1-\frac{k}{b}}}\psi\bra{x_1-\frac{k}{b}}dk \\
	&= \frac{b}{2\pi} \av{ \langle\psi,\widetilde{\psi}\rangle }^2.
\end{align*}

\medskip
Let us now prove \eqref{eq com(K_psi,K_tildepsi)=0 iff ...}. Let $\psi,\tilde{\psi}\in L^2(\RR)$ be real-valued and normalized. Let $F=\cW_b(f,g)$ for some normalized $f$ and $g$ in $L^2(\RR)$. Let $(g_k)_k$ be an orthonormal basis of $L^2(\RR)$ such that $g_{0}=g$.  One has
\begin{align*}
\bK_{\psi}^{2d} F &= \sum_k \langle F, \cW_b(\psi,g_k)\rangle \cW_b(\psi,g_k)\\
&= \sum_k \langle f,\psi\rangle\, \langle g, g_k \rangle \cW_b(\psi,g_k) \\
&= \langle f,\psi\rangle \, \cW_b(\psi,g).
\end{align*}
Therefore,
\begin{align*}
\bK_{\widetilde{\psi}}^{2d}\bK_{\psi}^{2d} F &= \langle f,\psi\rangle \, \langle\psi,\widetilde{\psi}\rangle \cW_b(\widetilde{\psi},g).
\end{align*}
Similarly,
\[ \bK_{\psi}^{2d}\bK_{\widetilde{\psi}}^{2d} F = \langle f,\widetilde{\psi}\rangle \, \langle\psi,\widetilde{\psi}\rangle \cW_b(\psi,g).
\]
It is then easy to see that if $\langle\psi,\widetilde{\psi}\rangle=0$ or $\widetilde{\psi}=\pm\psi$, then $\com{\bK_{\widetilde{\psi}}^{2d},\bK_{\psi}^{2d}}=0$. Conversely, assume that $\langle\psi,\widetilde{\psi}\rangle\neq 0$ and $\widetilde{\psi}\neq\pm\psi$. Then, one can find $f_0\in L^2(\RR)$ such that $\langle f_0,\psi\rangle=0$ and $\langle f_0,\widetilde{\psi}\rangle\neq 0$. For $F_0:=\cW_b^{2d}(f_0,g)$, we get $\bK_{\widetilde{\psi}}^{2d}\bK_{\psi}^{2d} F_0=0$ and $\bK_{\psi}^{2d}\bK_{\widetilde{\psi}}^{2d} F_0 \neq 0$, which ends the proof. 

\subsection{Proof of Proposition~\ref{prop F(b,rho)=...} (2d homogeneous electron gas)}\label{sect proofs-2}
Let $\eta\in\bS(L^2(\RR^2))$ such that $0\le\eta\le1$, and $[\eta,\fm_\bR^{b}]=0$, for all $\bR\in\RR^2$. By Theorem~\ref{thm decomposition of 2d operators for B orthogonal}, we may write $\eta$ as
$$
\eta= \sum_j \lambda_j\bK_{\psi_j}^{2d},
$$
where $\bK_{\psi_j}^{2d}$ is the orthogonal projector onto  
\[
E_{\psi_j}^{2d} = \left \{ \cW_{b}^{2d}(\psi_j,g):\ g\in L^2(\RR)\right \},
\]
and $\{\psi_j\}_{j\in\mathbb{N}}\subset L^2(\RR)$ is an orthonormal basis. 
We recall that 
$$
\bL^{2d}_\bA=\sum_{n\in\NN_0}\varepsilon_n \bM_n,
$$
 where $\bM_n:=\bK_{\phi_n^{b}}^{2d}$ is the $n$-th Landau projector, $\phi_n:=\phi_n^{b}$ refers to the Hermite Gauss function and $\varepsilon_n:=\varepsilon_n^b=b(2n+1)$ (see \eqref{eq def of Landau spectrals epsilon_n, phi_n}).
Then, according to \eqref{eq Tr(K_psi K_tildepsi)=...}, we have 
\begin{align*}\label{eq:EK}
\trv_2\bra{\bL_\bA^{2d}\eta}&= \sum_{j,n}\lambda_j\epsilon_n \trv_2\bra{\bM_n\bK_{\psi_j}^{2d}}\\
&=\frac{b}{2\pi} \sum_{j,n}\epsilon_n \alpha_{j,n}\lambda_j=\frac{b}{2\pi} \sum_n \epsilon_n m(n), 
\end{align*}
where $\alpha_{j,n}=\av{\langle \phi_n,\psi_j\rangle }^2$ and  $m(n)=\sum_j \alpha_{j,n} \lambda_j$.
Notice that $0\leq \lambda_j\leq 1$ and  $\sum_j \alpha_{j,n}=\sum_n \alpha_{j,n}=1$. 
Thus, $0\leq m(n)\leq 1$ and $\sum_n m(n)=\sum_j\lambda_j=2\pi\rho/b$. 
Now, by the bathtub principle \cite[Theorem 1.14]{LL}, one has 
$$
\inf\set{ \sum_n \epsilon_n m(n),\; 0\leq m(n)\leq 1,\; \sum_n m(n)=\frac{2\pi \rho}{b}}=  \sum_n \epsilon_n m^*(n),
$$
where 
$$
m^*(n)=\left\{ \begin{array}{ll}
		1 & \text{if } 0\leq n\leq \com{\frac{2\pi \rho}{b}}-1\\
		\set{\frac{2\pi \rho}{b}}&  \text{if } n= \com{\frac{2\pi \rho}{b}} \\
		0 & \text{otherwise}. 
\end{array}\right.
$$
Thus, a straightforward computation yields 
\begin{align*}
\inf&\set{ \sum_n \epsilon_n m(n),\; 0\leq m(n)\leq 1,\; \sum_n m(n)=\frac{2\pi \rho}{b}} \\ 
&\qquad \qquad \qquad= \frac{4\pi}{b} \bra{\pi\rho^2+\frac{b^2}{4\pi}\set{\frac{2\pi\rho}{b}}\bra{1-\set{\frac{2\pi\rho}{b}}}}.
\end{align*}
The claimed result follows.
\subsection{Proof of Proposition~\ref{prop omega_B(rho)=...} (3d homogeneous electron gas)} \label{sect proofs-3}
Let $\rho>0$ be constant and let $\gamma\in\bS(L^2(\RR^3))$ be such that $0\le\gamma\le1$ and $[\gamma,\fm_\bR^{b}]=0$, for all $R\in\RR^3$, with $\rho_\gamma(x_1,x_2,x_3)=\rho$. One has $\fm_{(0,0,R_3)}^{b}=\tau_{(0,0,R_3)}$. Hence, one can write
\[ \cF_3 \gamma \cF_3^{-1} = \int_\RR^\oplus \gamma_{k_3} \rd k_3, \]
 where $\gamma_{k_3}\in \bS(L^2(\RR^2))$ and $0\le \gamma_{k_3}\le 1$, for every $k_3$. Moreover, one has
 \begin{align*}
 \frac{1}{2}\VTr_{3}\left(\bL_\bA^{3d}\gamma\right) = \frac{1}{4\pi}\int_\RR \VTr_2\left( (\bL_\bA^{2d}+k_3^2)\gamma_{k_3}\right) \rd k_3.
 \end{align*}
 Moreover, since $[\gamma_{k_3},\fm_\bR^{b}]=0$, for every $k_3\in\RR$ and every $\bR\in\RR^2$, then, according to Theorem~\ref{thm decomposition of 2d operators for B orthogonal}, one can write
 $$\gamma_{k_3}=\sum_j \lambda_j(k_3)\bK_{\psi_j(k_3)}^{2d},$$
  with appropriate $(\psi_j(k_3))_j\subset L^2(\RR^2)$ and $\lambda_j(k_3)\in[0,1]$. With the same notation as in the proof of Proposition~\ref{prop F(b,rho)=...}, one has
  \[ (\bL_\bA^{2d}+k_3^2)\gamma_{k_3}= \sum_{n,j} \left(\epsilon_n+k_3^2\right)\lambda_j(k_3)  \bM_n \bK^{2d}_{\psi_j(k_3)}.  \]
 Thus,
 \[  \VTr_2\left((\bL_\bA^{2d}+k_3^2)\gamma_{k_3}\right) = \frac{b}{2\pi}\sum_{n,j} \left(\epsilon_n+k_3^2\right)\lambda_j(k_3) |\langle \phi_n,\psi_j(k_3)\rangle|^2.
 \]
We denote by $m(n,k_3):=\sum_j \lambda_j(k_3) |\langle \phi_n,\psi_j(k_3)\rangle|^2$. We then have
\[ \frac{1}{2}\VTr_{3}\left(\bL_\bA^{3d}\gamma\right) = \frac{b}{2(2\pi)^2}\int_\RR \sum_n \left(\epsilon_n+k_3^2\right) m(n,k_3) \rd k_3.
\]
As in the proof of Proposition~\ref{prop F(b,rho)=...}, we have $0\le m(n,k_3)\le 1$ and 
\begin{align*}
\int_\RR \sum_n m(n,k_3) &= \int_\RR \sum_n \sum_j \lambda_j(k_3) |\langle \phi_n,\psi_j(k_3)\rangle|^2 \rd k_3\\
&=  \int_\RR \sum_j \lambda_j(k_3) \rd k_3 = \frac{2\pi}{b} \int_\RR \VTr_2(\gamma_{k_3}) \rd k_3 \\
&= \frac{(2\pi)^2}{b}\VTr_{3}(\gamma)=\frac{(2\pi)^2 \rho}{b}.
\end{align*}
Then,
\begin{align}
\omega^{3d}(b,\rho)=\frac{b}{2(2\pi)^2}\inf\bigg \{ & \int_\RR \sum_n \left(\epsilon_n+k_3^2\right) m(n,k_3) \rd k_3 : 0\le m(n,k_3)\le 1, \; \\ \nonumber &  \int_\RR \sum_n m(n,k_3) \rd k_3=\frac{(2\pi)^2 \rho}{b}
\bigg\}.
\end{align}
Once again, the bathtub principle ensures that the above infinimum is obtained for $m^*(n,k_3)=\1_{\{\epsilon_n+k_3^2<\delta\}}$, for some positive $\delta>0$. Now, the constraint 
\[ \int_\RR \sum_n m^*(n,k_3) \rd k_3=\frac{(2\pi)^2 \rho}{b}
\]
yields \eqref{eq rho function of delta}. Moreover,
\begin{align*}
\omega^{3d}(b,\rho) &=\frac{b}{2(2\pi)^2}\int_\RR \sum_n \left(\epsilon_n+k_3^2\right) m^*(n,k_3) \rd k_3\\
&=\frac{b}{4\pi^2}\sum_n \epsilon_n\left(\delta-\epsilon_n\right)_+^{1/2} + \frac{b}{12\pi^2}\sum_n \left(\delta-\epsilon_n\right)_+^{3/2}\\
&= \frac{b}{6\pi^2}\sum_n \epsilon_n\left(\delta-\epsilon_n\right)_+^{1/2} +\frac{\delta \rho}{6}.
\end{align*}

\subsection{Proof of Theorem~\ref{thm Energy reduction in G} (3d electronic system with 2d symmetry)}\label{sect proofs-4}

\hspace{150mm}

We start by stating and proving a useful result. 

\begin{proposition}\label{prop vTr(L gamma)=}
	Let $\gamma\in\bS(L^2(\RR^3))$ be a locally trace class operator such that $[\gamma,\fm_\bR^{b_3}]=0$, for all $\bR\in\RR^2$. Let $(\lambda_n)_n\subset\RR_+^\NN$ and $(\psi_n)_n\subset L^2(\RR^2)$ be such that $\gamma=\sum_n \lambda_n \bK_{\psi_n}^{3d}$, with $\bK_{\psi_n}^{3d}$ as in Theorem~\ref{thm decomposition of 3d operators for B orthogonal}. One has
	\begin{equation}
		\VTr_2(\bL_\bA^{3d}\gamma)=\frac{b_3}{2\pi}\sum_n \lambda_n \langle \cH_{2d}\psi_n,\psi_n\rangle=\frac{b_3}{2\pi}\Tr(\cH_{2d}\gamma_0),
	\end{equation}
	where $\cH_{2d}=(\bp_1+b_2x_3)^2+\bp_3^2+b_3^2x_1^2$ has been introduced in~\eqref{eq anisotropic H-2d}, and $\gamma_0=\sum\lambda_n|\psi_n\rangle \langle \psi_n|$
	is the 
	zero-fiber of $\gamma$ through the partial Fourier transform in the $x_2$-direction.
\end{proposition}

\begin{proof}[Proof of Proposition~\ref{prop vTr(L gamma)=} ]

We start by pointing out the following identity that can be obtained by a straightforward calculation
\[ \bL_\bA^{3d} \cW^{3d}_{b_3}(f,g)= \cW^{3d}_{b_3}(\cH_{2d} f,g),\quad \forall (f, g) \in \cS(\RR^2)\times\cS(\RR).
\]
Now, let $\gamma=\sum_n \lambda_n \bK_{\psi_n}^{3d}$, with $\bK_{\psi_n}^{3d}=\1_{\set{\cW^{3d}_{b_3}(\psi_n,g)\; :\; g\in L^2(\RR)}}$ and let $(\psi_n)_n$ be an orthonormal basis of $L^2(\RR^2)$. One has
\begin{align*}
\bL_\bA^{3d}  \bK_{\psi_n}^{3d} &= \sum_j\av{  \bL_\bA^{3d} \cW^{3d}_{b_3}(\psi_n,g_j) \rangle \langle \cW^{3d}_{b_3}(\psi_n,g_j)}\\
&= \sum_j\av{  \cW^{3d}_{b_3}(\cH_{2d}\psi_n,g_j) \rangle \langle \cW^{3d}_{b_3}(\psi_n,g_j)}
\end{align*}
for any  $n\in\NN_0$ and $(g_j)$ an orthonormal basis of $L^2(\RR)$. Hence,
\begin{align*}
\rho_{\bL_\bA^{3d}\bK_{\psi_n}}(\bx)&=\sum_j \overline{\cW^{3d}_{b_3}(\psi_n,g_j)(\bx)} \cW^{3d}_{b_3}(\cH_{2d}\,\psi_n,g_j)(\bx)\\
&= \sum_j \overline{\langle \Phi_{\psi_n(\cdot,x_3),\bx}, g_j \rangle_{L^2(\RR)}} \langle \Phi_{\cH_{2d}\psi_n (\cdot,x_3),\bx}, g_j \rangle_{L^2(\RR)}\\
&= \langle \Phi_{\psi_n(\cdot,x_3),\bx}, \Phi_{\cH_{2d}\psi_n(\cdot,x_3),\bx}\rangle_{L^2(\RR)}\\
&=\frac{b_3}{2\pi}\langle\psi_n(\cdot,x_3),\cH_{2d}\,\psi_n\,(\cdot,x_3)\rangle_{L^2(\RR)},
\end{align*}
for each $n\in\NN_0$. Therefore, $$\VTr_2(\bL_\bA^{3d}\bK_{\psi_n}^{3d})=\int_\RR \rho_{\bL_\bA^{3d}\bK_{\psi_n}^{3d}}(x_3) \rd x_3 = \frac{b_3}{2\pi}\langle\psi_n , \cH_{2d}\psi_n\rangle_{L^2(\RR^2)},\quad\forall n\in\NN_0.$$
The claim now follows summing up over $n$.

\end{proof}

\begin{proof}[Proof of Theorem~\ref{thm Energy reduction in G}] 

Let $\gamma\in\cK$ and let $0\le\lambda_n\le 1$ and $(\psi_n)_n$ such that $\gamma=\sum_n \lambda_n\bK^{3d}_{\psi_n}$. By Proposition~\ref{prop vTr(L gamma)=}, we have
\[ \VTr_2(\bL_\bA^{3d} \gamma)=\frac{b_3}{2\pi}\sum_n \lambda_n \langle \psi_n, \cH_{2d}\psi_n\rangle.
\]  
We have shown in 
Section~\ref{app:anisotropic-harmonic-oscillator} that 
\begin{equation}\label{eq:L2d}
	\cH_{2d}= \mathcal{V}^{-1}\widetilde{\cH}_{2d}\mathcal{V}=\mathcal{V}^{-1}\cF_3 \bra{\int_\RR^\oplus h(k) dk} \cF_3^{-1}\mathcal{V},
\end{equation}
with $\mathcal{V}$ the multiplication operator by $e^{ib_2x_1x_3}$, $\widetilde{\cH}_{2d}= \bp_1^2+b_3^3x_1^2+(\bp_3-b_2x_1)^2$ and 
\begin{align*}
	h(k)&= -\partial_1^2+b_3^2x_1^2+(b_2x_1-k)^2\\
	&= -\partial_1^2+ |\bB|^2\bra{x-\frac{b_2}{|\bB|^2}k}^2 + \frac{b_3^3}{|\bB|^2}k^2.
\end{align*}
One has
\begin{align}\label{eq:hk}
h(k)&=\tau_{\frac{b_2}{|\bB|^2}k}\cH_{|\bB|}  \tau_{\frac{b_2}{|\bB|^2}k}^{-1}+ \frac{b_3^2}{|\bB|^2}k^2\nonumber \\
&= \sum_{m\in \NN_0}\bra{\varepsilon_m^{|\bB|} +\frac{b_3^2}{|\bB|^2}k^2}\av{\tau_{\frac{b_2}{|\bB|^2}k}\phi_m^{|\bB|} \Big \rangle \Big \langle \tau_{\frac{b_2}{|\bB|^2}k}\phi_m^{|\bB|}}. 
\end{align}
We now use~\eqref{eq:hk} and~\eqref{eq:L2d} to compute
	\begin{align*}
	\langle \psi_n, \cH_{2d}\psi_n\rangle= \Big\langle \bra{\int_\RR h(k)dk}\cF_{3}^{-1}\mathcal{V}\psi_n, \cF_{3}^{-1}\mathcal{V}\psi_n \Big\rangle. 
\end{align*}
Denoting by $\widetilde{\psi}_n(x_1,k)=(\cF_{3}^{-1}(\mathcal{V}\psi_n)(x_1,\cdot))(k) $, we have 
	\begin{align*}
	\langle \psi_n,\cH_{2d}\psi_n \rangle&= \int_\RR\langle   \widetilde{\psi}_n(\cdot,k) , h(k)\widetilde{\psi}_n(\cdot,k)\rangle dk\\
&= \int_\RR   dk \sum_{m\in \NN_0} \bra{\varepsilon_m^{|\bB|}+\frac{b_3^2}{|\bB|^2}k^2 } \av{\langle \phi_m^{|\bB|} (\cdot -\frac{b_2k}{|\bB|^2}), \widetilde{\psi}_n(\cdot,k) \rangle}^2 
\end{align*}
Let $c_{nm}(x_3)= \cF\set{ k\mapsto \langle \phi_m^{|\bB|} (\cdot -\frac{b_2k}{|\bB|^2}), \widetilde{\psi}_n(\cdot,k) \rangle }(x_3)$. Then 
	\begin{align*}
	\langle \psi_n, \cH_{2d}\psi_n\rangle&= \sum_m \varepsilon_m^{|\bB|} \norm{c_{nm}}_2^2+ \frac{b_3^2}{|\bB|^2} \norm{\nabla c_{nm}}_2^2 
\end{align*}
and 
$$
\VTr_2(\bL_\bA^{3} \gamma)=\frac{b_3}{2\pi}\sum_n \lambda_n\sum_m \varepsilon_m^{|\bB|} \norm{c_{nm}}_2^2+ \frac{b_3^2}{|\bB|^2} \norm{\nabla c_{nm}}_2^2.  
$$
Now, we define
\begin{equation}\label{eq loc def of G-gamma}
G_\gamma:=\frac{b_3}{2\pi}\sum_m\gamma_m,\;\text{where}\quad\gamma_m:=\sum_n \lambda_n |c_{n,m}\rangle\langle c_{n,m}|,\; \forall m\in\NN.
\end{equation}
Then,
\begin{equation}\label{eq:LAgamma-LAG}
\VTr_2(\bL_\bA^{3} \gamma)=\frac{b_3}{2\pi}\sum_m \varepsilon_m^{|\bB|} \Tr(\gamma_m)+\frac{b_3^2}{|\bB|^2}\Tr(-\Delta G_\gamma). 
\end{equation}
We want to use the bathtub principle as in~\cite{GLM23} to bound the RHS of~\eqref{eq:LAgamma-LAG} from below by a functional depending only on $G$. We start by proving some properties of $c_{nm}$ and $\gamma_m$. 
We have 
\begin{align*}
	\norm{c_{nm}}_2^2=\norm{\check{c}_{nm}}_2^2
	&= \int dk \av{ \int \overline{\phi}_m (x-\frac{b_2k}{|\bB|^2}) \widetilde{\psi}_n(x,k) dx}^2\\
	&= \int dk \av{ \int \overline{\phi}_m (x) \widetilde{\psi}_n(x+\frac{b_2k}{|\bB|^2},k) dx}^2.
\end{align*}
Thus 
\begin{align*}
	\sum_m 	\norm{c_{nm}}_2^2= \int dk  \int \av{\widetilde{\psi}_n(x+\frac{b_2k}{|\bB|^2},k)}^2 dx= \int dk \int \av{\widetilde{\psi}_n(x,k)}^2 dx= \norm{\widetilde{\psi}_n}_2^2=1. 
 \end{align*}
Besides, $0\leq \gamma_m\leq 1$ for any $m$. Indeed, for $f\in L^2(\RR)$
\begin{align*}
	\langle f,\gamma_m f\rangle &= \sum_n\lambda_n \av{\langle c_{nm},f\rangle}^2=\sum_n\lambda_n \av{\langle \check{c}_{nm},\check{f}\rangle}^2\\
	&=\sum_n\lambda_n\av{ \int \phi_m(x-b_2k/|\bB|^2)\overline{\widetilde{\psi}_n}(x,k)\check{f}(k)dx dk}^2\\
	&\leq \sum_n \av{\langle\widetilde{\psi}_n , (x,k)\mapsto \phi_m(x-b_2k/|\bB|^2)\check{f}(k)\rangle }^2. 
\end{align*}
As $(\widetilde{\psi}_n)$ is an orthonormal basis, then 
$$
	\langle f,\gamma_m f\rangle \leq \norm{(x,k)\mapsto \phi_m(x-b_2k/|\bB|^2)\check{f}(k)}_2^2= \norm{f}_2^2.
$$
Furthermore, $\rho_\gamma(x_3)=\rho_{G_\gamma}(x_3)$. Indeed, one has
\begin{align*}
c_{n,m}(x_3) &= \frac{1}{\sqrt{2\pi}}\int_\RR \int_\RR \re^{-\ri k x_3}\overline{\varphi_m^{|\bB|}}\left(x_1-\frac{b_2}{|\bB|^2}k\right)\widetilde{\psi}_n(x_1,k)\, \rd k\, \rd x_1\\
&= \frac{1}{\sqrt{2\pi}}\int_\RR \int_\RR \re^{-\ri k x_3}\overline{\varphi_m^{|\bB|}}(x_1)\widetilde{\psi}_n\left(x_1+\frac{b_2}{|\bB|^2}k,k\right)\, \rd k\, \rd x_1.
\end{align*}
Since $(\varphi_m^{|\bB|})_m$ forms an orthonormal basis of $L^2(\RR)$. Then,
\begin{align}\label{eq loc int psi^2=sum c_nm}
\sum_m |c_{n,m}(x_3)|^2 &= \frac{1}{2\pi}\int_\RR \left|\int_\RR \re^{-\ri k x_3}\widetilde{\psi}_n\left(x_1+\frac{b_2}{|\bB|^2}k,k\right)\, \rd k \right|^2 \rd x_1 \nonumber\\
&= \frac{1}{2\pi}\int_\RR \left|\int_\RR \re^{-\ri k x_3}\widetilde{\psi}_n\left(x_1,k\right)\, \rd k \right|^2 \rd x_1 \nonumber\\
&= \int_\RR \left| \psi_n\left(x_1,x_3\right) \right|^2 \, \rd x_1. 
\end{align}
Therefore, using \eqref{eq rho_gamma(x_3)=...}, one gets
\begin{align}\label{eq loc rho-G=rho-gamma}
\rho_{G_\gamma}(x_3) &= \frac{b_3}{2\pi} \sum_m \rho_{\gamma_m}(x_3) = \frac{b_3}{2\pi} \sum_m \sum_n \lambda_n |c_{n,m}(x_3)|^2 \nonumber\\
&= \frac{b_3}{2\pi} \sum_n \int_\RR \lambda_n |\psi_n(x_1,x_3)|^2 \, \rd x_1 =\rho_\gamma(x_3).
\end{align}
Finally, if we write the spectral decomposition of $G_\gamma$ as $G_\gamma=\sum_j \mu_j |g_j\rangle\langle g_j|$. One has $\Tr(G_\gamma)=\sum_j\mu_j$ and evaluating $\Tr(\gamma_m)$ in the basis $(g_j)_j$ one obtains
\[ \Tr( \gamma_m)= \sum_j\langle g_j,\gamma_m g_j\rangle.
\]
Notice that, for every $j$, $\sum_m \langle g_j,\gamma_m g_j\rangle = \frac{2\pi}{b_3}\langle g_j,G_\gamma g_j\rangle=\frac{2\pi}{b_3}\mu_j$, and $0\le\langle g_j,\gamma_m g_j\rangle \le1$, for every $j,m$. Therefore, from~\eqref{eq:LAgamma-LAG}
\[ \VTr_2(\bL_\bA^{3d} \gamma)\ge \frac{b_3}{2 \pi}\sum_j I_j +\frac{b_3^2}{\av{\bB}^2}\Tr(-\Delta G_\gamma),\]
where
\[ I_j:=\inf\set{ \sum_m \varepsilon_m^{|\bB|} f_j(m) : 0\le f_j(m)\le1\;\text{and}\; \sum_m f_j(m) =\frac{2\pi\mu_j}{b_3}  }.
\]
Similarly to the proof of Proposition~\ref{prop F(b,rho)=...}, one can conclude by the bathtub principle that $I_j= \sum_m \varepsilon_m^{|\bB|} f_j^*(m)$, where
\begin{equation}\label{eq loc f-j^*(m)=...}
	f_j^*(m)=\left\{ \begin{array}{ll}
		1 & \text{if } 0\leq m\leq \com{\frac{2\pi \mu_j}{b_3}}-1\\
		\set{\frac{2\pi \mu_j}{b_3}}&  \text{if } m= \com{\frac{2\pi \mu_j}{b_3}} \\
		0 & \text{otherwise}. 
	\end{array}\right.
\end{equation}
This yields $I_j=\frac{4\pi\av{\bB}}{b_3^2}\omega^{2d}(b_3,\mu_j)$. Summing up over $j$, we obtain
\begin{equation}\label{eq loc Tr(LA gamma) > ...}
	\VTr(\bL_\bA^{3d} \gamma)\ge 2\frac{\av{\bB}}{b_3}\Tr(\omega^{2d}(b_3,G_\gamma)) +\frac{b_3^2}{\av{\bB}^2}\Tr(-\Delta G_\gamma).
\end{equation}
Hence,
\begin{equation}\label{eq:b2neq0}
\inf_{\gamma \in \cK \atop \rho_\gamma = \rho} \set{ \frac{1}{2}\VTr(\bL^\bA_3\gamma)}\ge\inf_{ G\in \cG \atop \rho_G = \rho} \set{ \frac{1}{2}\frac{b_3^2}{\av{\bB}^2}\Tr(-\Delta G)+\frac{\av{\bB}}{b_3}\Tr(\omega^{2d}(b_3,G))}.   
\end{equation}

\bigskip

In order to obtain an equality in the above inequality \eqref{eq:b2neq0}, we shall  assign to any $G\in\cG$, an operator $\gamma\in\cK$ such that $\rho_\gamma=\rho_G$ and so that there is equality in \eqref{eq loc Tr(LA gamma) > ...}. To do so, let $G=\sum_j \mu_j\,|g_j\rangle\langle g_j|\in\cG$ and set $\gamma=\sum_{j,m} \lambda_{j,m}\bK_{\psi_{j,m}}^{3d}$, with $\lambda_{j,m}$ is as in \eqref{eq loc f-j^*(m)=...}, and $\set{\psi_{j,m}}\subset L^2(\RR^2)$ will be constructed suitably. Notice that
$0\le\lambda_{j,m}\le 1$ and 
$$ \VTr_2(\gamma)=\frac{b_3}{2\pi}\sum_{j,m}\lambda_{j,m}= \sum_j \mu_j=\Tr(G).$$
Furthermore, one has by Proposition~\ref{prop vTr(L gamma)=}
\begin{align*}
\VTr_2(\bL_\bA^{3d} \gamma)&=\frac{b_3}{2\pi}\sum_{j,m}\lambda_{j,m}\langle \psi_{j,m}, \cH_{2d} \psi_{j,m}\rangle 
\end{align*}
and, as previously,
\begin{align*}
  \langle \psi_{j,m}, \cH_{2d} \psi_{j,m}\rangle =\int_\RR \langle \tilde{\psi}_{j,m}(\cdot,k), h(k)\widetilde{\psi}_{j,m}(\cdot,k)\rangle ,
\end{align*}
where $\widetilde{\psi}_{j,m}=\cV^*\cF_3 \psi_{j,m}$. We choose
$$ \widetilde{\psi}_{j,m}(x_1,k)=\phi_m^{|\bB|}\left(x_1-\frac{b_2}{\bB|^2}k\right)\,\hat{g}_j(k),
$$ 
so that 
\begin{align*}
\langle  \psi_{j,m}\cH_{2d}\, \psi_{j,m}\rangle &=\int_\RR \left( \varepsilon_m^{|\bB|}+\frac{b_3^2}{|\bB|^2} k^2  \right)\av{\hat{g}_j(k)}^2\, \rd k\\
&= \varepsilon_m^{|\bB|}+\left(\frac{b_3}{|\bB|}\right)^2 \|\nabla g_j\|_2^2.
\end{align*}
It thus follow that
\begin{align*}
\VTr_2(\bL_\bA^{3d} \gamma)&=\frac{b_3}{2\pi}\sum_{j,m}\lambda_{j,m} \varepsilon_m^{|\bB|} +\left(\frac{b_3}{|\bB|}\right)^2 \frac{b_3}{2\pi}\sum_{j,m}\lambda_{j,m}\|\nabla g_j\|_2^2\\
&= 2\frac{|\bB|}{b_3} \sum_{j}\omega^{2d}(b_3,\mu_j)+\left(\frac{b_3}{|\bB|}\right)^2 \sum_j \mu_j \|\nabla g_j\|_2^2.
\end{align*}
Finally, one concludes that
\begin{align*}
    \frac{1}{2}\VTr_2(\bL_\bA^{3d} \gamma)= \frac{|\bB|}{b_3}\Tr(\omega^{2d}(b_3,G))+\frac{1}{2}\left(\frac{b_3}{|\bB|}\right)^2\Tr(-\Delta G).
\end{align*}
To complete the proof, we need to show that $\rho_G=\rho_\gamma$. This follows by the same computations as in \eqref{eq loc int psi^2=sum c_nm} and \eqref{eq loc rho-G=rho-gamma}.

\end{proof}

\appendix

\section{Asymptotic behavior of $\omega^{3d}(b,\rho)$}\label{appendix}

We detail in this appendix the proof of Proposition~\ref{prop:asymptot-w3d} which concerns the behavior of the kinetic energy density $\omega^{3d}(b,\rho)$, for weak magnetic fields $b\ll1$, and large magnetic fields $b\gg1$. 

\subsection{Behavior near $0$}
We are going to prove the convergence
\begin{equation}\label{eq lim omega_B (rho)=...}
\lim_{b\to0}\omega^{3d}(b,\rho)=\frac{\pi^{4/3}}{6^{1/3}}\rho^{5/3}.
\end{equation}
Let $b>0$ and $N_b:=\lfloor \left(\frac{\delta}{b}-1\right)/2\rfloor$. According to \eqref{eq omega_B(rho)=}, one has

\begin{align*}
\omega^{3d}(b,\rho)=\frac{\delta\rho}{6}+S_b \ge \frac{\delta\rho}{6},
\end{align*}
where $$ S_b = \frac{b^2}{6\pi^2} \sum_{n=0}^{N_b}\varepsilon_n^b (\delta-\varepsilon_n^b)_+^{1/2} .$$
Besides,
\begin{align*}
0\le S_b =& \dfrac{b^{\frac{7}{2}}}{6\pi^2} \sum_{n=0}^{N_b}(2n+1)\left(\frac{\delta}{b}-1-2n\right)^{1/2} \\
\le& \dfrac{\sqrt{2} b^{\frac{7}{2}}}{6\pi^2} \sum_{n=0}^{N_b}(2n+1)\left(N_b-n+1\right)^{1/2}\\
=&  \dfrac{\sqrt{2} b}{6\pi^2}\left(b(N_b+1)\right)^{5/2}\frac{1}{N_b+1}\sum_{n=0}^{N_b} \frac{2n}{N_b+1}\left(1-\frac{n}{N_b+1}\right)^{1/2} \\
& + \dfrac{\sqrt{2} b^2}{6\pi^2}\left(b(N_b+1)\right)^{3/2}\frac{1}{N_b+1}\sum_{n=0}^{N_b} \left(1-\frac{n}{N_b+1}\right)^{1/2}.
\end{align*}
Recall that $\delta$ is defined as the unique real number satisfying
\begin{align*}
	  \sum_{n=0}^{N_b}  \left(\delta -b (2n+1)\right)^{1/2}=\frac{2\pi^2\rho}{b}.
\end{align*}
Let $f:t\mapsto \sum_{n}\bra{t-n}_+^{1/2}$, so that $\delta= b\bra{2f^{-1}\bra{\frac{2\pi^2\rho}{b^{3/2}}}+1}$. $f$ and $f^{-1}$ are increasing coercive functions, therefore, the behavior of $\delta$ as $b\to 0$ is dictated by the behavior of $f$ at infinity.  
 For $n=\lfloor t\rfloor$, we have
 $$A_n:=f(n)=\sum_{k=0}^n\sqrt{k}\leq f(t)< f(n+1)=A_{n+1}.$$
Besides, 
$$
A_n=n^{3/2}\bra{ \frac{1}{n}\sum_{k=0}^n\sqrt{\frac{k}{n}}}= n^{3/2}\bra{\frac23+ O\bra{\frac{1}{n}}}. 
$$
Therefore 
$f(t)\simeq \frac{2}{3}t^{3/2}$, $f^{-1}(y)=\bra{\frac{3}{2}y}^{2/3}$ and $\delta\simeq 2(3\pi^2)^{2/3}\rho^\frac23+b$. It follows that $N_b= O(1/b)$ and $b(N_b+1)=O(1)$ as $b\to0$. Thus, 
\[ \frac{1}{N_b+1}\sum_{n=0}^{N_b} \frac{2n}{N_b+1}\left(1-\frac{n}{N_b+1}\right)^{1/2} \longrightarrow \int_0^1 2x\sqrt{1-x}\, \rd x=\frac{8}{15},\quad \text{as}\quad b\to 0
\]
and
\[ \frac{1}{N_b+1}\sum_{n=0}^{N_b} \left(1-\frac{n}{N_b+1}\right)^{1/2} \longrightarrow \int_0^1 \sqrt{1-x}\, \rd x=\frac{2}{3},\quad \text{as}\quad b\to 0.
\]
This shows that $S_b\to 0$ as $b\to0$. On the other hand, 
\[ \lim_{b\to0}\omega^{3d}(b,\rho)=\lim_{b\to 0}\frac{\delta\rho}{6}=\frac{2(3\pi^2)^{2/3}}{6}\rho^{5/3}=\frac{(3\pi^2)^{2/3}}{3}\rho^{5/3}.
\]
\subsection{Behavior at $\infty$} For $b>(2\pi^4\rho^2)^{1/3}$, equation \eqref{eq rho function of delta} becomes $(\delta-b)^{1/2}=\frac{2\pi^2\rho}{b}$, thus
\[ \delta =b+\left( \frac{2\pi^2\rho}{b}\right)^2.\]
Therefore, \eqref{eq omega_B(rho)=} becomes
\begin{align*}
    \omega^{3d}(b,\rho)&=\left( b+\frac{4\pi^4\rho^2}{b^2}\right)\frac{\rho}{6}+\frac{b^3}{6\pi^2}\frac{2\pi^2\rho}{b}\\
    &=\left(2b^2+b\right)\frac{\rho}{6}+\left(\frac{2\pi}{b}\right)^2\frac{\rho^3}{6}= \frac{b^2\rho}{3}+O(b).
\end{align*}
%
%
%


\bibliographystyle{plain}

\bibliography{fichierb} 

\end{document}